\documentclass[3p,times]{elsarticle}
\usepackage{graphics}
\usepackage{graphicx}
\usepackage{caption}
\usepackage{subcaption}
\usepackage{amssymb}
\usepackage{bm}
\usepackage{amsmath}
\usepackage{color}
\usepackage{multirow}
\usepackage{amsthm}
\usepackage{multicol}
\usepackage{mwe}
\usepackage{float}
\usepackage{enumerate}

\newtheorem{theorem}{Theorem}

\newtheorem{proposition}[theorem]{Proposition}
\newtheorem{lemma}{Lemma}
\newtheorem{corollary}{Corollary}

\newtheorem{remark}{Remark}

\def\bP{\mathbb{P}}
\def\bE{\mathbb{E}}
\def\bR{\mathbb{R}}
\def\bN{\mathbb{N}}
\def\cF{\mathcal{F}}

\usepackage{bbm}
\usepackage{placeins}

\DeclareMathOperator{\cov}{Cov}

\DeclareMathOperator{\Var}{Var}

\usepackage{hyperref}
\hypersetup{
colorlinks=true
}
\usepackage{epstopdf}

\biboptions{authoryear}

\makeatletter
\def\ps@pprintTitle{%
 \let\@oddhead\@empty
 \let\@evenhead\@empty
 \def\@oddfoot{\centerline{\thepage}}%
 \let\@evenfoot\@oddfoot}
\makeatother

\begin{document}

\begin{frontmatter}

\title{Gaussian dependence structure pairwise goodness-of-fit testing based on conditional covariance and the 20/60/20 rule}

\author[label1]{Jakub Wo\'zny}

\author[label2]{Piotr Jaworski}

\author[label1]{Damian Jelito\corref{cor1}} \ead{Damian.Jelito@uj.edu.pl}
\author[label1]{Marcin Pitera}

\author[label3]{Agnieszka Wy\l oma\'nska}

\address[label1]{Institute of Mathematics, Jagiellonian University, S. {\L}ojasiewicza 6, 30-348 Krak{\'o}w, Poland}
\address[label2]{Institute of Mathematics, University of Warsaw, S. Banacha 2, 02-097 Warsaw, Poland}
\cortext[cor1]{Corresponding author.}
\address[label3]{Faculty of Pure and Applied Mathematics, Hugo Steinhaus Center, Wroc{\l}aw University of Science and Technology, Hoene-Wronskiego 13c, \\50-376 Wroc{\l}aw, Poland}

\begin{abstract}
We present a novel data-oriented statistical framework that assesses the presumed Gaussian dependence structure in a pairwise setting. This refers to both multivariate normality and normal copula goodness-of-fit testing. The proposed test clusters the data according to the 20/60/20 rule and confronts conditional covariance (or correlation) estimates on the obtained subsets. The corresponding test statistic has a natural practical interpretation, desirable statistical properties, and asymptotic pivotal distribution under the multivariate normality assumption. We illustrate the usefulness of the introduced framework using extensive power simulation studies and show that our approach outperforms popular benchmark alternatives. Also, we apply the proposed methodology to commodities market data.
\end{abstract}
\begin{keyword}
correlation \sep covariance \sep multivariate normality test \sep Gaussian copula \sep normal copula \sep  factor copula model \sep 20/60/20 rule \sep Markowitz portfolio optimisation \sep Incremental Risk Charge \sep Default Risk Charge
\end{keyword}
\end{frontmatter}

\section{Introduction}
Second-moment matrices are a common way to describe, characterise, or quantify dependence in the multivariate setting, see \cite{DroKot2001}, \cite{Mui2009}, and references therein. For example, this refers to factor copula models or optimal portfolio selection processes, which are often based on covariance/correlation matrix specifications, see, e.g., \cite{Car2009b} and \cite{WilBruCho2013}. From the modelling perspective, the usage of a particular dependence structure, either covariance-based or not, should be verified using statistical methods. In fact, a statistical procedure, commonly known as a {\it goodness-of-fit testing}, is typically performed to assess the suitability of the presumed dependence structure and its consistency with the available data. In the literature, there are multiple types of goodness-of-fit statistical tests focused on distribution or copula fit, including blanket tests, family-class-specific tests, and likelihood ratio tests. We refer to \cite{BerBak2006}, \cite{GenRemBea2009}, \cite{Fer2013}, and references therein for a review of the available methods with particular emphasis on the dependence structure testing. 
 
In practice, one should be careful when using covariance/correlation as a dependence structure fit metric, since linear dependence measures might fail to adequately capture various nonlinear characteristics. This problem is extensively discussed in the literature and is known to result in various modelling challenges. We refer to \cite{Sti1989} and \cite{Ald1995} for the historical background, and to \cite{Vel2015}, \cite{EmbMcnStr2002}, \cite{ZeeMas2002}, \cite{ZhaOkhZhoSon2016}, and \cite{TjoOtnSt02022} for a more recent analysis of properties and pitfalls resulting from the use of second-moment matrices for dependence characterisation. To mitigate this, at least from a theoretical perspective, the usage of the covariance/correlation matrix as a sole dependence indicator is often accompanied by explicit or implicit normality or ellipticity assumption, see \cite{PafKon2004} and \cite{GreLau2004}. As expected, this assumption is later used in the targeted goodness-of-fit testing, which facilitates efficient dimension reduction (e.g. through stochastic radius evaluation) and reliance on a well-known univariate distribution discrimination test statistics, see \cite{BenBon1980}, \cite{JasHAuMin2017}, \cite{AmeSen2020}, and references therein. 

In this context, despite its well-known flaws, the usage of multivariate normal distribution or normal copula is arguably the most common and pragmatic modelling choice, see e.g. \cite{DAS2024105310}, \cite{PENG2022104940}. Let us alone mention Markowitz's portfolio optimisation models or the Incremental Risk Charge (IRC) models used to quantify rating transition risk within the Market Risk capital framework, which are often based, directly or indirectly, on the normality setup, see \cite{McnFreEmb2010}. This is also reflected in the rich literature on targeted multivariate normality goodness-of-fit testing. We refer to \cite{EbnHen2020} for an excellent recent survey on this topic; see also \cite{henze2002invariant},  \cite{mecklin2004appraisal} and the references therein. Note that under the normality null hypothesis, the targeted normality goodness-of-fit tests are known to have better discriminatory power when compared to blanket tests focused, e.g., on generic distribution distance measurement; this also applies to normal copula goodness-of-fit setup, as the copula could be easily transferred to the multivariate normal distribution via the Rosenblatt normality transform, see \cite{MelBec2018} and \cite{MalSor2003}. That being said, to achieve the pivotal quantity property, the majority of powerful goodness-of-fit normality test statistics lack natural practical (non-mathematical) interpretation and they are often based on nontrivial normalisation schemes (utilising inverse empirical covariance matrix) and dimension reduction techniques, which motivates the need for further research in this area.

In this paper, as a potential remediation for the aforementioned problems, we introduce an intuitive, straightforward, and direct pairwise covariance-based goodness-of-fit test that can be used to assess the adequacy of multivariate normality or the normal copula distributional assumption. Our test is based on a very simple idea: we check whether the quantile-conditioned covariances on specific subsets of data are aligned. 
The new test could be seen as a multivariate dependence-oriented extension of the statistical framework developed in \cite{JelPit2021}, where conditional second-moment analysis was used to construct a powerful univariate normality testing framework based on the 20/60/20 rule. In the present setting, we focus on dependence assessment and assume that the marginal distribution fit is adequate (or is checked separately); this is roughly equivalent to the copula goodness-of-fit testing, see Remark~\ref{rem:margins}; we refer to~\cite{MalSor2003,GenRemBea2009,AmeSen2020} for a similar setting. See also \cite{SCHEPSMEIER201534} for a comprehensive review.

The key tool in the study is the conditional covariance-based test statistic $T_n$ introduced in Equation~\eqref{T.stat}. In a nutshell, it measures the empirical difference between the conditional covariance values under the 20/60/20 population split induced by a benchmark random variable defined in Equation~\eqref{eq:hatY}; this is used to establish a pairwise normality (normal copula) goodness-of-fit statistical test. The theoretical properties of $T_n$ are stated in Theorem~\ref{theorem}, which is the key theoretical result of the paper. We show that the test statistic introduced in this paper is an asymptotic pivotal quantity with the standard normal asymptotic distribution; this also indirectly shows that the conditional covariance distance measurement process leads to consistent estimates. In the second part of the paper, we compare the performance of $T_n$ against other benchmark goodness-of-fit alternatives. This is done via statistical power analysis using the Monte Carlo approach under popular alternative model specifications. Our analysis indicates that despite its simplicity, the test statistic $T_n$ often has the best discriminatory power, at least when faced with heavy-tailed dependence structure alternatives.

For completeness, it should be noted that our framework and test statistic design are based on a series of theoretical results that provide convenient analytical formulas of covariance/correlation matrices under the normality assumption. First, it has been shown in \cite{JawPit2015} that if a multivariate random vector $X$ is normally distributed, then conditional covariance matrices on specific subsets are equal to each other irrespective of the choice of the underlying parameter. Namely, if we introduce a benchmark random variable $Y$ that is a linear combination of the margins of $X$, and then split the population according to the values of $Y$ using a ratio close to 20/60/20, then the conditional covariance matrices of $X$ are equal to each other. Second, it was shown that a similar setup, but with potentially different ratios, is true for the elliptical case, see \cite{JawPit2017}. Third, conditional covariance/correlation matrices have been identified as efficient distribution classifiers, see \cite{JawPit2020}. Finally, the conditional correlation matrices could be used to obtain the complete characterisation of independence, see~\cite{JawJelPit2023}. Taken together, those results show that the comparison of conditional covariance/correlation matrices could be used to detect non-Gaussian behaviour, especially when it is linked to an increase of dependence in the tails; see \cite{JelPit2021} for details. For simplicity, we focus on bivariate dependency testing, but our framework effectively refers to (pairwise) multivariate setting, see Section~\ref{S:application2}.

This paper is organised as follows. In Section~\ref{S:preliminaries} we set up the terminology and notation linked to conditional covariance matrices and correlation matrices, and recall the results linking the conditional covariance framework to multivariate normality. Then, in Section~\ref{S:application2} we introduce a distribution characteristic used for the pairwise goodness-of-fit test, which is introduced later in Section~\ref{SS:Statistic}. In Section~\ref{S:power} we evaluate the performance of the introduced test using Monte Carlo simulations and in Section~\ref{S:empirical} we apply the proposed methodology to empirical data. Technical lemmas and proofs are deferred to~\ref{appendix.A}.

\section{Preliminaries}\label{S:preliminaries}
Let $(\Omega,\cF,\bP)$ be a probability space. Throughout this paper we use $X=(X_1,\ldots,X_k)^\top$ to denote a non-degenerate $k$-dimensional (column) random vector with finite second moments, i.e. $X_i\in L^2(\Omega,\cF,\bP)$, for $i=1,\ldots,k$, and fixed $k\in\bN\setminus \{0\}$. For a fixed {\it loading factor} vector $\alpha=(\alpha_1,\ldots,\alpha_k)^\top\in\mathbb{R}^k\setminus \{0\}$, we define a linear combination of margins of $X$, induced by $\alpha$, by setting
\begin{equation}\label{eq:Y.alpha}
Y:=\alpha_1X_1+\ldots+\alpha_kX_k,
\end{equation}
and refer to $Y$ as a {\it benchmark} of $X$; note that using the scalar product $\langle \cdot,\cdot\rangle$ notation, we may simply write $Y=\langle \alpha,X\rangle$. Next, for any given {\it quantile split} values $a,b\in [0,1]$, where $a<b$, we define a benchmark based {\it quantile conditioning set} given by 
\begin{equation}\label{eq:A.alpha}
A(a,b):=\{y\in\bR\colon F^{-1}_{Y}(a)<y<F^{-1}_{Y}(b)\},
\end{equation}
where $F^{-1}_{Y}$ denotes the (generalised) inverse of the cumulative distribution function (CDF) of $Y$; sometimes we also write $A_Y(a,b)$ to emphasise the dependence of $A(a,b)$ on $Y$. Note that the set $A(a,b)$ is scale-invariant, i.e. if we replace $Y$ with $\lambda Y$, for some $\lambda>0$, then $A(a,b)$ will not change. Because of that, the {\it loading factor} $\alpha\in \bR^k\setminus \{0\}$ could be parameterised using $k-1$ parameters corresponding e.g. to a normalised $\alpha$ satisfying $\Vert \alpha\Vert_{2} =1$. Still, for clarity, we decided to use the original notation. Next, if the choice of $Y$ as well as quantile splits $a$ and $b$ is fixed, we use the simplified notation $A:=A(a,b)$ and set
\begin{align*}
\mu_{A}& :=\mathbb{E}[X\mid Y\in A],\\
\cov[X\mid Y\in A] & :=\mathbb{E}[(X-\mu_{A})(X-\mu_{A})^\top \mid Y\in A],
\end{align*}
to denote the $A$-set conditional mean vector and conditional variance-covariance matrix, respectively. Note that these objects are well defined, since $\bP[Y\in A]= b-a>0$ and $X\in L^2(\Omega,\cF,\bP)$; the latter assumption is, in fact, only required for a limit situation when $a=0$ or $b=1$. 
With a slight abuse of notation, we sometimes write $\cov[X | A]$ instead of $\cov[X | Y\in A]$. For consistency, we use $\mu$ and $\cov[X]$ (or $\Sigma$) to denote the unconditional mean and variance-covariance matrix of $X$, respectively.  

The conditional variance-covariance matrix enjoys a very natural interpretation. Similarly to the unconditional variance-covariance matrix, it could be seen as a dependence measure of $X$, given that we subset the data based on $Y$. In particular, in the financial context, one can assume that $X$ describes the rates of returns of some assets and $Y$ is the rate of return of the associated index. Then, $\cov[X|Y\in A]$ measures the dependence between the returns given that the overall state of the market corresponds to the event ${Y\in A}$. For instance, the stressed market condition, where the index returns are particularly low, could be modelled by low values of $a$ and $b$, and the corresponding conditional covariance matrix described the dependence of $X$ in this extreme regime.

As usual, we use $\Phi$ to denote the CDF of the standard normal random variable and write $X\sim N_k(\mu,\Sigma)$ if we want to assume that $X$ follows a multivariate normal distribution with mean $\mu$ and covariance matrix $\Sigma$. 

\subsection{Quantifying the dependence structure:  copula function and conditional covariance matrices}

In this paper, we focus on the assessment of dependence structure. Consequently, to streamline the narrative, we assume that the choice of marginal distributions of $X$ is adequate and does not need to be checked. To better explain this assumption, we provide more comments about a situation where the margin modelling is split from the dependence modelling. Given a generic random variable $Z$, let $F_Z$ and $F^{-1}_Z$ denote its CDF and its (generalised) inverse, respectively. Also, let $H$ denote the multivariate CDF of the random vector $X$. Using this notation and Sklar's Theorem, we know that $H$ could be expressed as
\[
H(t_1,\ldots,t_k)=C(F_{X_1}(t_1),\ldots,F_{X_k}(t_k)),\qquad \, (t_1,\ldots,t_k)^\top\in\bR^k,
\]
where $C\colon [0,1]^k\to [0,1]$ denotes the copula function of $X$, see Theorem 2.10.9 in~\cite{nelsen2007introduction} for details. Then, we can simply say that we assume that the marginal CDFs of $X$ are known (up to affine transforms), but the choice of the underlying copula function $C$ requires verification. Moreover, we assume that the dependence structure is induced by the covariance matrix in the sense that the function $C$ could be effectively expressed as a function of $\cov [X]$; note that the parametrisation could be equivalently based on the correlation matrix since copulas are invariant to (positive) affine transformations of the margins. For example, this could refer to $C$ belonging to the class of normal or t-student copulas. This is a common setup within multi-factor copula models or portfolio optimisation frameworks, see Remark~\ref{rem:margins} and Remark~\ref{rem:IRC} in this paper and Section 11.5 in \cite{Hul2018}. Note that if we assume that all margins of $X$ are standard normal, then covariance matrix $\Sigma$ is equal to the correlation matrix and the copula of $X$ is given by $C(u_1,\ldots,u_k)=H(\Phi^{-1}(u_1),\ldots,\Phi^{-1}(u_k))$, see~\cite{nelsen2007introduction} for details.

Let us now comment on what is the interaction between the pre-assumed dependence structure encoded in $H$ and conditional covariance matrices. First, from \cite{JawPit2020}, we know that the law of the benchmark $Y$ can be recovered (up to an additive constant) if we are given the values of conditional covariance matrices for all quantile splits. Indeed, it suffices to note that the conditional variance of $Y$ on the set $A_Y(a,b)$ can be recovered directly from $\cov[X\mid Y\in A_Y(a,b)]$ and use Theorem~3.1 from \cite{JawPit2020}. In particular, this implies that we can fully recover the multivariate distribution of $X$ (up to additive shift), if we are given conditional covariance matrices for all splits and loading factors. For completeness, we formulate this result in Theorem~\ref{prop.Jaw2} which could be seen as a special case of Theorem 3.2 in \cite{JawPit2020}; we omit the proof for brevity.

\begin{theorem}[Conditional covariance matrices as distribution classifiers]\label{prop.Jaw2}
    Let $X$ and $\tilde X$ be random vectors such that
    \[
    \cov[X\mid Y\in A_Y]=\cov[\tilde X\mid\tilde Y\in A_{\tilde Y}],
    \]
    for any $0\leq a< b\leq 1$ and $\alpha\in \bR^{k}\setminus\{0\}$, where $Y=\langle \alpha,X\rangle$, $\tilde Y=\langle \alpha,\tilde X\rangle$, $A_Y=\left\{y\in\bR\colon F^{-1}_Y(a)< y< F^{-1}_Y(b) \right\}$, and $A_{\tilde Y}=\left\{y\in\bR\colon F^{-1}_{\tilde Y}(a)\leq  y\leq F^{-1}_{\tilde Y}(b) \right\}$. Then, there exists $c\in\bR^k$ such that $H_X(t)=H_{\tilde X+c}(t)$ for $t\in\bR^{k}$, i.e. the multivariate distributions of $X$ and $\tilde X$ coincide almost surely up to an additive shift. In particular, the copulas of $X$ and $\tilde X$ coincide.
\end{theorem}
\noindent On the other hand, the usage of a single loading factor might be insufficient to recover the law of $X$. This can be easily checked by setting $Y=X_1$ and assuming $X_1$ is independent of the vector $(X_2,\ldots,X_k)$. Thus, while the set of all conditional covariance matrices is a distribution classifier, the single conditional covariance matric is its characteristics. In next section, we show how to use conditional covariances to construct efficient characteristics that measure increase of dependancy in the tail sets, which are in turn used for efficient normal distribution discrimination.

\begin{remark}[Marginal transforms and copulas]\label{rem:margins}
Assuming that the margins of $X$ are known, we might effectively assume that they are standard normal or could be transferred to standard normal using e.g. the Rosenblatt normality transform, see \cite{MelBec2018}. In fact, this is true for any (not necessarily normal) type of marginal distribution. Because of that, we often assume that the margins are consistent with the chosen copula family in a way that it facilitates a direct estimation of copula parameters based on the covariance (correlation) structure.
\end{remark}

\begin{remark}[Correlation matrix as a dependence structure characteristic]\label{rem:IRC}
Usage of the covariance (or correlation) matrix as a sole dependence measure is typically accompanied by additional assumptions. Within the portfolio optimisation setup, where $X$ is used to describe future asset returns, this is often accompanied by assuming a specific dependence structure imposed on $X$, using a multi-factor model, or using variance as a key risk indicator, see \cite{OweRab1983,EltGruBrogoe2014,Kou2020}. Furthermore, there are many models in which the correlation dependence structure is imposed indirectly. In particular, this refers to Incremental Risk Charge (IRC) and Default Risk Charge (DRC) models, which are often based on multi-factor copula models with correlation structure adequacy assumption (encoded in the copula function choice) being a key model aspect, see e.g. \cite[Section III.C]{EBA2012} or \cite[Section 6.3]{EGIM}. In particular, while the model output might be linked to a discrete-type outcome (e.g. changes in credit ratings), the dependence structure might be still calibrated using continuous-type data (e.g. asset correlations, spreads) which constitute $X$, see \cite{WilBruCho2013}, \cite{MarLutWeh2011}, \cite{GreLau2004}, and references therein.
\end{remark}

\subsection{Conditional covariance matrices and normality assumptions: the 20/60/20 rule}\label{S:CCV}

While a closed form formula for the conditional covariance matrix $\cov[X|A]$ is hard to obtain in general, it can be directly derived for the multivariate normal distributions. For completeness, we state it in Proposition~\ref{pr:cov.normal}; for the proof, see Theorem 4.1 and Example 4.2 in~\cite{JawPit2017}. 

\begin{proposition}\label{pr:cov.normal}
Let $X\sim N_k(\mu,\Sigma)$ and $\alpha\in\bR^k\setminus\{0\}$. Then, for any $0\leq a<b\leq 1$ and $A=A(a,b)$, we get
\begin{equation}\label{eq:covX.normal}
\cov[X\mid Y\in A]=\cov[X]+\left(\Var[Y\mid Y\in A]-\Var[Y]\right)\beta\beta^\top,
\end{equation}
where $\beta:=\cov[Y,X]/\Var[Y]=(\alpha\Sigma\alpha^\top)^{-1}\Sigma \alpha^\top$ is the vector of regression coefficient from orthogonal projection of margins of $X$ onto $Y$.
\end{proposition}

From Proposition~\ref{pr:cov.normal} we infer that, under the normality assumption, it is relatively easy to derive $\cov[X\mid Y\in A(a,b)]$ for any quantile split pair $(a,b)$. Note that the dependence on $A(a,b)$ is in fact expressed only through the conditional variance of the benchmark, i.e. the value $\Var[Y\mid Y\in A(a,b)]$. Consequently, if $\Var[Y\mid Y\in A(a,b)]=\Var[Y\mid Y\in A(a',b')]$ for two quantile split pairs $(a,b)$ and $(a',b')$, then we get $\cov[X\mid A(a,b)]=\cov[X\mid A(a',b')]$. This simple observation could be used to recover {\it the $20/60/20$ rule} for the multivariate normal distribution. The rule states that if we split the probability space based on the values of benchmark $Y$ and follow the ratio close to 20/60/20, then the corresponding conditional covariance matrices of $X$ are equal to each other, see Theorem~\ref{Th:206020}.

\begin{theorem}[The 20/60/20 Rule]\label{Th:206020}
Let $X\sim N_k(\mu,\Sigma)$ and $\alpha\in\bR^k\setminus\{0\}$.  Then, we get
\begin{equation}\label{eq:206020}
\cov[X\mid Y\in A_1]=\cov[X\mid Y\in A_2]=\cov[X\mid Y \in A_3],
\end{equation}
where $A_1:=A(0,\tilde{q})$, $A_2:=A(\tilde{q},1-\tilde{q})$, and $A_3:=A(1-\tilde{q},1)$, for $\tilde{q}:=\Phi(\tilde{x})\approx 0.19808$, and $\tilde{x}$ being a unique positive solution to the equation $-x \Phi(x)-\phi(x)(1-2 \Phi(x))=0$.
\end{theorem}

The split ratio $\tilde q / (1-2\tilde q)/ \tilde q$ presented in Theorem~\ref{Th:206020} is roughtly equal to 20/60/20 and is invariant, i.e. it is true for all values of the underlying parameters $\mu$, $\Sigma$, and $\alpha$, see \cite{JawPit2015} for more details and proofs.

To illustrate this phenomenon, in Figure~\ref{F:1} we present 20/60/20 population splits in the bivariate normal setting under different benchmark and correlation specifications. For any choice of $\alpha$ and the (unconditional) correlation $\rho$, the sample conditional covariance matrices for the subsamples (indicated by different colours) are close to each other.

\begin{figure}[htp!]
\centering
\includegraphics[width=0.9\textwidth]{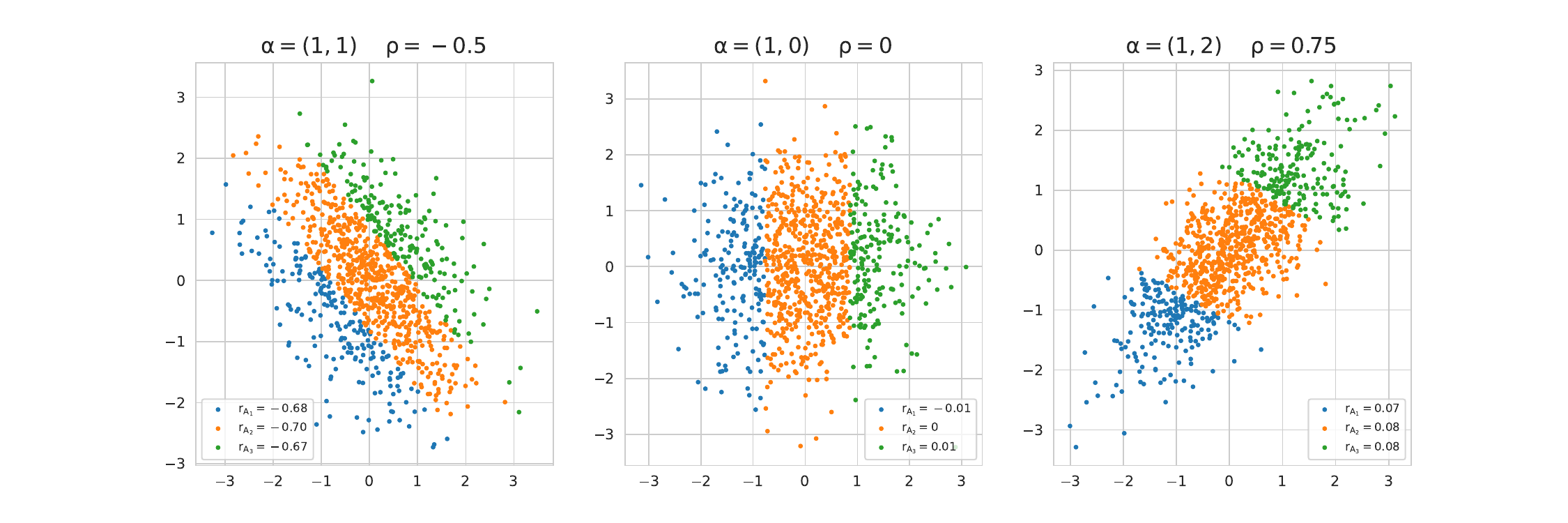}
\caption{The plot presents conditioned populations of $X\sim N_k(\mu,\Sigma)$ for $k=2$, $\mu=(0,0)^T$, and $\Sigma$ having unit diagonals and non-diagonal entry $\rho\in [1,1]$. Conditioning is induced by sets $A_1$, $A_2$, and $A_3$, for three exemplary values of $\alpha\in [0,1]^2\setminus\{0\}$ and $\rho\in [-1,1]$. The empirical bivariate covariance matrices for 20/60/20 ratio conditioned samples are close to each other and asymptotically equal.}\label{F:1}
\end{figure}
\FloatBarrier


\section{Pairwise conditional covariance and 20/60/20 rule induced equilibrium}\label{S:application2}

In the rest of the paper, we focus on a simplified bivariate ($k=2$) setting and compare conditional covariance values for a fixed pair of random variables. Rather than operating on copula functions, we work on generic distributions with pre-assumed normal margins; the margins parameters are assumed to be unknown, see Remark~\ref{rem:margins}. This means that we are effectively performing a bivariate normality goodness-of-fit testing assuming normality of margins.

Our setting corresponds to a more practical environment in which we want to cross-check each individual input of the correlation matrix separately; note that each unconditional covariance also depends only on the corresponding pair of random variables. Note that while this is a simplified testing setup in which we assess bivariate (sub-copula) building blocks, it should allow us to reduce problem dimension while effectively detecting local conditional covariance divergences. To be more precise, from now on we set $k=2$, and consider $X=(X_1,X_2)$ with
\begin{equation}\label{eq:Y.2dim}
Y=\alpha_1X_1+\alpha_2X_2,
\end{equation}
for some fixed $\alpha\in\bR^2\setminus\{0\}$. For any quantile conditioning set $A=A(a,b)$,\,$0\leq a<b\leq 1$, we use a simplified notation 
\[
\mu_A=(\mu_{A,1},\mu_{A,2})\quad \textrm{and}\quad r_A=\cov[X_1,X_2\mid Y\in A],
\]
to denote $A$-set conditional mean and conditional covariance, respectively. Similarly, we use $\mu$ and $r$ to denote their unconditional versions. Furthermore, we denote unconditional marginal variances of $(X_1,X_2)$ by $\sigma_1^2$ and $\sigma_2^2$, respectively, and the correlation between the margins by $\rho$. Recalling \eqref{eq:206020}, we know that if $X$ is multivariate normal, then we get $r_{A_1}=r_{A_2}=r_{A_3}$, 
where $A_1$, $A_2$, and $A_3$ are defined in Theorem~\ref{Th:206020} and correspond (approximately) to a 20/60/20 population split based on the benchmark value. From this we immediately get that the 20/60/20 rule induced {\it equilibrium} equality
\begin{equation}\label{eq:covariance.eq}
r_{A_1}-2 r_{A_2}+r_{A_3}=0,
\end{equation}
holds if $X$ is multivariate normal. Consequently, by studying the sample-based estimate of $r_{A_1}-2 r_{A_2}+r_{A_3}$, one can detect dependence structures which are different than normal, see \cite{JelPit2021} where a similar idea in the univariate setting was considered. In particular, the bigger the size of this value, the bigger the tail dependence measured by the conditional tail covariance.

From \eqref{eq:Y.2dim}, we see that the benchmark random variable $Y$ is defined separately for each pair of margin random variables. To assess the dependence structure of the initial $k$-dimensional vector, say $\tilde X=(\tilde X_1,\ldots,\tilde X_k)$, one might fix benchmark proportions and perform $\frac{k(k-1)}{2}$  tests for each pair of margins of $\tilde X$. As an alternative, one can also consider a generic benchmark, say $\tilde Y:=\tilde a_1\tilde X_1+\ldots\tilde a_k\tilde X_k$, project the initial vector to a bivariate space, and then perform more targeted testing. For instance, we can perform $k$ individual tests which assess dependence structure between the benchmark and each margin. Namely, for $i\in \{1,\ldots,k\}$, we can set $X=(X_1,X_2)$ with $X_1=\tilde X_i$ and $X_2=\tilde Y$, and perform the standard pairwise testing; note that if $\tilde X$ is multivariate normal, then so is $X$. In this setup the usage of marginal loading factor $a=(0,1)$ leads to $\tilde Y=Y$, i.e. we can consider non-altered quantile conditioning set in the pairwise testing.

\begin{remark}[Bivariate normality goodness-of-fit testing]\label{rem:bivariate.testing}
Note that while \eqref{eq:covariance.eq} is designed to focus on dependence structure characterisation, we can use the 20/60/20 rule to also assess the marginal fit. This can be done e.g. by incorporating the methodology from \cite{JelPit2021}, where univariate (margin) normality goodness-of-fit test based on conditional standard deviations and 20/60/20 rule was introduced. Considered together, this would relate to the assessment of the distance between the whole conditional covariance matrices rather than only its diagonal entries, see also Remark~\ref{rem:multidimensional}. That saying, for simplicity, in this paper we decided to focus on covariance structure assessment.
\end{remark}

\begin{remark}[Multivariate goodness-of-fit testing]\label{rem:multidimensional}
While in this paper we focus on pairwise testing, one can consider the $k$-dimensional setting and use the 20/60/20 rule to compare the whole $k\times k$ dimensional covariance matrices directly, using a single benchmark. For instance, this could be done via a formula similar to \eqref{eq:covariance.eq} given by
\begin{equation}\label{eq:test.multi}
\cov[X\mid Y\in A_1]-2\cov[X\mid Y\in A_2]+\cov[X\mid Y\in A_3].
\end{equation}
Namely, we expect that the sample-based estimate of \eqref{eq:test.multi} should be close to zero in any matrix norm. That saying, the derivation of the asymptotic distribution of the corresponding test statistic might be a highly challenging task as one would need to control the statistical error encoded in the whole $k$-dimensional sample vector. Thus, this approach is left for future research.
\end{remark}

\section{Pairwise goodness-of-fit test based of conditional covariance and the 20/60/20 rule}\label{SS:Statistic}

In this section we formally introduce the test statistic associated with~\eqref{eq:covariance.eq}. We start with setting up the statistical framework. Let $n\in\bN\setminus \{0\}$ and assume  that $((X_{i,1},X_{i,2}))_{i=1}^{n}$ is an independent and identically distributed sample from $X$. As in~\eqref{eq:Y.2dim}, we fix $\alpha=(\alpha_1,\alpha_2)^\top \in\mathbb{R}^2\setminus\{0\}$ and define the {\it benchmark sample} $(Y_i)_{i=1}^{n}$ given by
\begin{equation}\label{eq:hatY}
Y_i:=\alpha_1X_{i,1}+\alpha_2X_{i,2},\quad i=1,2,\ldots,n.
\end{equation}
In this section, with a slight abuse of notation, we use "$(i)$" to denote the index corresponding to $i$th order statistic of the sample $(Y_i)_{i=1}^{n}$. In this setting, while $(Y_{(i)})_{i=1}^{n}$ is the standard order statistic vector, the vector $(X_{(i),1})_{i=1}^{n}$ is ranked according to the values of the benchmarked sample $(Y_i)_{i=1}^{n}$ rather than $(X_{i,1})_{i=1}^{n}$; same applies to  $(X_{(i),2})_{i=1}^{n}$. 
Given quantile split values $a,b\in [0,1]$, where $a<b$, we can consider the empirical (ranked) sub-sample corresponding to $A=A(a,b)$. This sub-sample is simply given by $((X_{(i),1},X_{(i),2}))_{i=[na]+1}^{[n b]}$, where $[x]:=\max\{k\in\mathbb{Z}:k\leq x\}$ is the integer part of $x\in\bR$. Let us now provide formulas for the sample estimators of the conditional moments of $X$ on the set $A$. Namely, for $j\in\{1,2\}$ and $A=A(a,b)$, where $0\leq a<b\leq 1$, by setting $m_n:=[nb]-[na]$,  we define the sample conditional mean and sample conditional covariance, respectively, by
\begin{align}
\textstyle\hat{\mu}_{A,j} &\textstyle :=\frac{1}{m_n}\sum_{i=[na]+1}^{[nb]}X_{(i),j}\label{eq:cond.estimators:mu},\\
\textstyle \hat{r}_A & \textstyle :=\frac{1}{m_n}\sum_{i=[na]+1}^{[nb]}\left(X_{(i),1}-\hat{\mu}_{A,1}\right)\left(X_{(i),2}-\hat{\mu}_{A,2}\right).\label{eq:cond.estimators}
\end{align}
Also, we use $\hat{\mu}_{j}$ and $\hat{r}$ to denote the unconditional versions of \eqref{eq:cond.estimators:mu} and \eqref{eq:cond.estimators}, respectively,  which could be recovered by simply setting $a=0$ and $b=1$. We also use $\hat{\sigma}^2_{1}$ and $\hat{\sigma}^2_{2}$
to denote the unconditional sample variance estimators of $X_1$ and $X_2$, respectively.

Now, we define the test statistic based on \eqref{eq:covariance.eq}. Let the sets $A_1$, $A_2$, and $A_3$ be given as in Theorem~\ref{Th:206020}, and relate to the 20/60/20 split. Then, we can define the normalised sample equivalent of \eqref{eq:covariance.eq} by setting
\begin{equation}\label{T.stat}
    T_n:=\frac{\sqrt{n}}{\hat{\tau}}\left(\hat{r}_{A_1}-2\hat{r}_{A_2}+\hat{r}_{A_3}\right),
\end{equation}
where $\hat{\tau}$ is a normalising statistic which aim is to estimate the standard deviation of $\sqrt{n}(\hat{r}_{A_1}-2\hat{r}_{A_2}+\hat{r}_{A_3})$. Namely, we set
\begin{equation}\label{c.tau}
    \hat{\tau}^2:=\left(\frac{\hat{r}_{Y,1}\hat{r}_{Y,2}}{\hat{\sigma}^2_Y}\right)^2K_1 + \left(\frac{\hat{r}_{Y,1}^2\hat{\sigma}^2_2+2\hat{r}_{Y,1}\hat{r}_{Y,2}\hat{r}+\hat{r}_{Y,2}^2\hat{\sigma}^2_1}{\hat{\sigma}^2_Y} \right)K_2 + \left(\hat{\sigma}_1^2\hat{\sigma}_2^2+\frac{2\hat{r}_{Y,1}\hat{r}_{Y,2}\hat{r}}{\hat{\sigma}^2_Y}\right)K_3,
\end{equation}
where $\hat{r}_{Y,1}:=\alpha_1\hat{\sigma}^2_1+\alpha_2\hat{r}$, $\hat{r}_{Y,2}:=\alpha_1\hat{r}+\alpha_2\hat{\sigma}^2_2$, and  $\hat{\sigma}^2_Y:=\alpha_1^2\hat{\sigma}^2_1+\alpha_2^2\hat{\sigma}^2_2+2\alpha_1\alpha_2\hat{r}$ are the (unconditional) sample estimators of $\cov[Y,X_1]$, $\cov[Y,X_2]$,  and $\textrm{Var}[Y]$, respectively, and where $K_1\approx 22.0766$, $K_2\approx -29.8012$ and $K_3\approx 33.4424$ are fixed and determined by the 20/60/20 split, see Equation~\eqref{eq:th:final:variance} in \ref{appendix.A} for exact formulas.

Let us now explain the idea behind the definition of $T_n$. If the sample follows a bivariate normal distribution, the differences between the estimates $\hat{r}_{A_1}$, $\hat{r}_{A_2}$, $\hat{r}_{A_3}$ should be close to zero due to Equation \eqref{eq:covariance.eq}. The test statistic essentially confronts the conditional linear dependence in the central region ($A_2$) with the conditional linear dependence in the tails ($A_1$ and $A_3$); note that in this bivariate setting, the term "tail" is defined with the help of the univariate benchmark sample $(Y_i)$. For most common alternative choices, it is expected that the theoretical property \eqref{eq:covariance.eq} would not hold due to increased or decreased tail dependence, which should have a direct impact on \eqref{T.stat}.  

The test statistic $T_n$ assumes (implicitly) distribution symmetry, by combining left- and right-tail covariances $\hat r_{A_1}$ and $\hat r_{A_3}$ in a single formula. To account for a possible alternative distribution asymmetry, in addition to~\eqref{T.stat}, we introduce two complementary test statistics
\begin{equation}\label{T1.stat}
    L_n:=\frac{\sqrt{n}}{\hat{\eta}}\left(\hat{r}_{A_1}-\hat{r}_{A_2}\right)\quad\quad\textrm{and}\quad\quad R_n:=\frac{\sqrt{n}}{\hat{\eta}}\left(\hat{r}_{A_3}-\hat{r}_{A_2}\right),
\end{equation}
where $\hat{\eta}$ is a modified version of \eqref{c.tau} given by
\begin{equation}\label{c.delta}
  \hat{\eta}^2:=\left(\frac{\hat{r}_{Y,1}\hat{r}_{Y,2}}{\hat{\sigma}^2_Y}\right)^2\tilde K_1+\left(\frac{\hat{r}_{Y,1}^2\hat{\sigma}^2_2+2\hat{r}_{Y,1}\hat{r}_{Y,2}\hat{r}+\hat{r}_{Y,2}^2\hat{\sigma}^2_1}{\hat{\sigma}^2_Y} \right)\tilde K_2+\left(\hat{\sigma}_1^2\hat{\sigma}_2^2+\frac{2\hat{r}_{Y,1}\hat{r}_{Y,2}\hat{r}}{\hat{\sigma}^2_Y}\right)\tilde K_3,
\end{equation}
where $\tilde K_1\approx 8.8484$, $\tilde K_2\approx -11.9491$, and $\tilde K_3\approx 13.4091$ are normalising constants; formula for $\hat{\eta}$ is obtained in the same way as formula for $\hat{\tau}$, see Remark~\ref{rm:constants2} in~\ref{appendix.A} for details. These test statistics focus on the right tail or left tail dependence and could be recommended if one expects some form of asymmetry in the alternative distribution.

Now, we state the main theoretical result of this paper, which identifies the asymptotic distribution of $T_n$ (as well as $L_n$ and $R_n$) under the null hypothesis that the copula $X$ is Gaussian; this asymptotic distribution is a pivotal quantity, i.e. it is invariant with respect to the choice of all underlying parameters $\mu$, $\Sigma$, and $\alpha$. 

\begin{theorem}\label{theorem}
    Let $X\sim N_2(\mu,\Sigma)$, $\alpha\in\mathbb{R}^2\setminus\{0\}$, and $T_n$ be given by \eqref{T.stat}. Then, we get
    \begin{equation}
        T_n\xrightarrow{d}N(0,1), \quad n\to \infty.
    \end{equation}
    Moreover, the same result is true for $L_n$ and $R_n$.
\end{theorem}

From Theorem~\ref{theorem} we also get the consistency of the non-normalised equilibrium estimator of \eqref{eq:covariance.eq}, i.e. we know that $\hat r_{A_1}-2\hat r_{A_2} +\hat r_{A_3}\to 0$ in probability, as $n\to\infty$.  For brevity, we defer the proof of Theorem~\ref{theorem} to \ref{appendix.A}, where we also show several properties of sample conditional moments which could be of independent interest. 

For illustration purposes, in Figure~\ref{fig:Distributions} we present the estimated density of $T_n$ for various sample sizes and parameter choices. As one can see, the convergence to the standard normal distribution is relatively fast and stable within the considered sets of parameters; note that various choices of $\rho$ indicate various choices of correlations as we assumed standardised margins.

\begin{figure}[htp!]
    \begin{center}

    \includegraphics[width=0.9\linewidth]{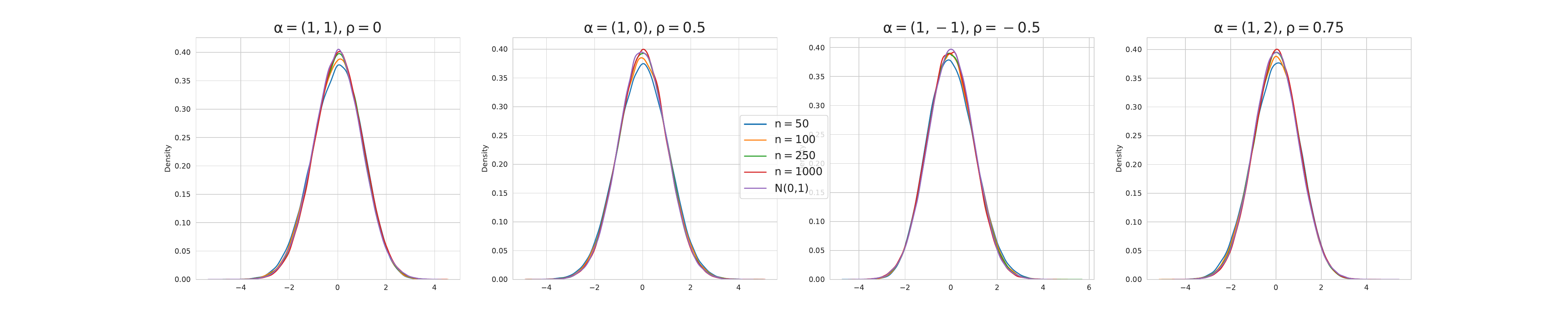}
    \end{center}
    \caption{Kernel density estimates of the distribution of $T_n$ under the normality assumption for different sample sizes and parameter choices. The results are based on 100,000 Monte Carlo samples of size $n$. As expected, all empirical densities are close to the standard normal density and the results are stable with respect to the underlying parameter choice.}
    \label{fig:Distributions}
\end{figure}

We conclude this section with three additional test statistics that could be seen as modifications of~\eqref{T.stat} and~\eqref{T1.stat} based on the decorrelation procedure. Namely, given an initial sample $(X_{i,1},X_{i,2})_{i=1}^n$, we introduce modified sample $(\tilde X_{i,1},\tilde X_{i,2})_{i=1}^n$, where
\[
\textstyle \tilde{X}_{i,1}:=\frac{X_{i,1}+X_{i,2}}{\hat\sigma(X_{1}+X_{2})}\quad\textrm{and}\quad \tilde{X}_{i,2}:=\frac{X_{i,1}-X_{i,2}}{\hat\sigma(X_{1}-X_{2})},\quad i=1, \ldots, n,
\]
for $\hat\sigma(X_{1}+X_{2})=(\hat\sigma_1^2+\hat\sigma_2^2+2\hat r)^{1/2}$ and $\hat\sigma(X_{1}-X_{2})=(\hat\sigma_1^2+\hat\sigma_2^2-2\hat r)^{1/2}$. By construction, the empirical covariance (correlation) of $(\tilde{X}_{i,1},\tilde{X}_{i,2})_{i=1}^n$ should always be close to zero since the margins of this sample have approximately unit variances. Given the modified sample $(\tilde X_{i,1},\tilde X_{i,2})_{i=1}^n$, we introduce a modified test statistic
\begin{equation}\label{T_tilde}
 \tilde T_n:= T_n((\tilde{X}_{i,1},\tilde{X}_{i,2})_{i=1}^n),
\end{equation}
where $T_n((\tilde{X}_{i,1},\tilde{X}_{i,2})_{i=1}^n)$ denotes test statistic $T_n$ applied to sample $(\tilde X_{i,1},\tilde X_{i,2})_{i=1}^n$. The test statistic $\tilde T_n$ might be seen as a stabilised version of the test statistic $T_n$ given in~\eqref{T.stat} in which the correlation effect is reduced. In particular, the definition of $\tilde{T}_n$ effectively induces data-dependent loading factors: If the benchmark $Y$ for the original test $T_n$ corresponds to the loading factors $\alpha=(1,1)$, then the benchmark for $\tilde T_n$ tests corresponds to the loading factors $
\tilde\alpha =( (\hat\sigma_1^2+\hat\sigma_2^2+2\hat r)^{-1/2}+(\hat\sigma_1^2+\hat\sigma_2^2-2\hat r)^{-1/2},(\hat\sigma_1^2+\hat\sigma_2^2+2\hat r)^{-1/2}-(\hat\sigma_1^2+\hat\sigma_2^2-2\hat r)^{-1/2})$ applied to the original sample; this depends on the estimated variances and the underlying empirical covariance.  As will be shown later, by introducing $\tilde T_n$, we get better stability of the test power, especially when the sample correlation is very high and margins strongly depend on each other, see Section~\ref{S:empirical} for details. Following a similar logic, we also introduce modifications of test statistics given in~\eqref{T1.stat}, by setting
\[
\tilde L_n:= L_n((\tilde{X}_{i,1},\tilde{X}_{i,2})_{i=1}^n)\quad\textrm{and}\quad \tilde R_n:= R_n((\tilde{X}_{i,1},\tilde{X}_{i,2})_{i=1}^n).
\]
It should be noted that by the argument similar to the one used in the proof of Theorem~\ref{theorem}, one can show the asymptotic normality of $\tilde{T}_n$, $\tilde{L}_n$, and $\tilde{R}_n$; the detailed proof is omitted for brevity.

\section{Power analysis}\label{S:power}
In this section, we assess the statistical power of the test statistics introduced in Section~\ref{SS:Statistic} and confront them with selected benchmark alternatives.  For simplicity, if not stated otherwise, we consider a bivariate setup with standard normal margins and varying dependence structure modelled by an appropriate copula function. 
Also, for the test statistics introduced in Section~\ref{SS:Statistic} we fix loading factors $\alpha=(1,1)$ which put equal weights to both coordinates and provide a generic check when no additional information is provided.

\subsection{Benchmark test statistics}\label{SS:benchmark}
In this section, we briefly introduce alternative benchmark statistics which are often used for multivariate normality testing. These statistics will be later included in the comparative statistical power assessment, where the power of the new procedures will be confronted with the power of benchmark statistics for various choices of alternative distributions. Namely, we consider four well-established benchmark statistics, i.e. Baringhaus–Henze–Epps–Pulley (BHEP), Anderson-Darling (AD), Cramer-von Mises (CM) and Mardia's skewness (MS) tests. They are proven to provide a comprehensive multivariate normality assessment for various choices of alternatives. For a more detailed discussion on the construction and properties of those statistics, we refer to a recent comprehensive survey paper \cite{EbnHen2020}, see also \cite{henze2002invariant,mecklin2004appraisal}. For completeness, we provide a brief description for each statistic in the generic $k$-dimensional setting. In order to achieve affine invariance, all the aforementioned statistics are based on the \emph{scaled residuals sample} $(Z_{i})_{i=1}^{n}$ given by
    $Z_{i}:=\hat \Sigma^{-\frac{1}{2}}(X_i-\hat \mu)$, $i\in \{1,2,\ldots,n\}$,
where $\hat\mu$ is the sample mean and $\hat \Sigma$ is the sample covariance matrix, which is (a.s.) symmetric and positively defined. 

The first benchmark, i.e. the BHEP test statistics, is based on the distance between the theoretical and empirical characteristic functions. Namely, we set 
\begin{equation}\label{eq:BHEP}
\textrm{BHEP}_{n,\beta}:=\int_{\mathbb{R}^d}|\hat \Psi(t)-\exp\left(-\tfrac{||t||^2}{2}\right)|^2\phi_\beta(t)\,dt,
\end{equation}
where $\hat \Psi(t):=\frac{1}{n}\sum_{i=1}^n \exp(\mathbf{i}t^T Z_{i})$ is the empirical characteristic function of $(Z_i)$, $\phi_\beta(t):=(2\pi\beta^2)^{-d/2}\exp(-\frac{||t||^2}{2\beta^2})$ is a  weight function depending on the smoothing parameter $\beta>0$, $\mathbf{i}$  stands for the imaginary unit, and $t\to \exp\left(-\frac{||t||^2}{2}\right)$ is a characteristic function of a standard multivariate distribution. In this paper, we take $\beta=1$. For details, we refer to \cite{henze1990class}.

The second benchmark, i.e. AD test statistics, is based on distance measurement between the theoretical and empirical distribution function of {\it square radii}; this relates to the distribution of the radius random variable $R$ in the stochastic elliptical representation, see~\cite{GomGomMar2003} for details. In a nutshell, we define sample $(R_i)_{i=1}^{n}$, where $R_{i}:=||Z_{i}||^2=Z_{i}Z_{i}^\top$, $i=1,\ldots,n$, and check if this sample follows the $\chi^2_k$ distribution which should be the case for the Gaussian distribution. Then, the AD test statistic is given by
\begin{equation}\label{eq:AD}
    \textrm{AD}_n:=\int_0^\infty\frac{n\left(\hat F_R(t)-F_{\chi^2_k}(t)\right)^2}{F_{\chi^2_k}(t)\left(1-F_{\chi^2_k}(t)\right)}\,dF_{\chi^2_k}(t),
\end{equation}
where $F_{\chi^2_k}$ is CDF of $\chi^2_k$, and $\hat F_R(t)
:=\frac{1}{n}\sum_{i=1}^n 1_{\{R_{i}\leq t\}}$, $t\in\bR$, is the sample distribution of $R$. Note that AD can be expressed in a simplified form as $\textrm{AD}_n=-n-\sum_{i=1}^{n}\left(\log F_{\chi^2_k}(R_{(i)})+\log(1-F_{\chi^2_k}(R_{(n+1-i)}))\right)$, where $R_{(i)}$ is the standard $i$th order statistic of $R$;
we refer to \cite{paulson1987some} for more details. 

The third benchmark, i.e. CM test statistics, is also based on {\it square radii} measurement but with a different weighting scheme. The CM statistic is given by
\[
\textrm{CM}_n:=\int_0^\infty n\left(\hat F_R(t)-F_{\chi^2_k}(t)\right)^2 \, dF_{\chi^2_k}(t),
\]
using the same notation as in the definition of AD test statistic. As before, this statistic could be also expressed in a simplified form as $\textrm{CM}_n=\frac{1}{12n}+\sum_{i=1}^n\left(F_{\chi^2_k}(R_{(i)})-\tfrac{2i-1}{2n}\right)^2$, see \cite{koziol1982class}. It should be noted that square radii distance-based type statistics were also considered in \cite{MalSor2003} and~\cite{GenRemBea2009}. The difference could be linked to a modified construction of radii and some additional data transformations (e.g. Rosenblatt transform or probability integral transform). Nevertheless, for simplicity, we decided to use only the present definition.

The fourth benchmark, i.e. MS test statistics, is essentially a measure of multivariate skewness which aims to detect whether the generalized skewness of the sample corresponds to the one induced by the multivariate normal distribution shape, i.e. if it is close to zero. The test statistic is given by 
\[
    \textrm{MS}_n:=\frac{1}{n^2}\displaystyle\sum_{i=1}^n \sum_{k=1}^n\left( Z_{i}^\top Z_{k}\right)^3,
\]
we refer to \cite{mardia1970measures} for more details.

\begin{remark}[Heavy-tailed and light-tailed hypothesis testing]\label{rm:fat_tail}
It should be noted that all test statistics discussed in this section could only be used with one-sided rejection regions, where large values of the test statistic lead to rejection of the null hypothesis. However, the proposed new test statistic~\eqref{T.stat} could be used with one-sided as well as two-sided rejection regions. More specifically, if it is expected that the dependence in the tails (measured by the conditional correlation) is bigger (resp. smaller) than the dependence in the central part of the sample, then the right-sided (resp. left-sided) rejection region could be used. This improves the flexibility of the proposed methodology and positively affects the statistical power. Also, note that imbalances in the tail dependence could be associated with heavy or light-tail phenomena with tailored asymmetric test statistics given in~\eqref{T1.stat}.    
\end{remark}


\subsection{Statistical power}\label{SS:power}

In this section we assess the statistical power of $T_n$, $R_n$, $\tilde T_n$ and $\tilde R_n$, and confront it with benchmark tests $\textrm{BHEP}_n$, $\textrm{AD}_n$, $\textrm{CM}_n$, and $\textrm{MS}_n$; note we do not present results for $L_n$ and $\tilde L_n$ as this could be easily transferred to $R_n$ and $\tilde R_n$ by a simple transformation of the data. The null hypothesis always corresponds to the bivariate normal distribution. For the alternative hypotheses, we choose various distributions with normal margins and non-normal copula functions; this choice is summarised in Table~\ref{table.copulas}.

\begin{table}[htp!]
\centering
\begin{tabular}{|l|l|l|}
\hline
Copula name      & Formula & Used parameters                  \\ \hline
Student's t & $C_{\nu,\rho}(u,v)=\frac{\Gamma\left(\frac{\nu+2}{2}\right)}{\Gamma\left(\frac{\nu}{2}\right)\sqrt{(1-\rho^2)(\pi\nu)^2}}\int_{-\infty}^{t^{-1}_\nu(u)}\int_{-\infty}^{t^{-1}_\nu(v)}\left(1+\frac{x_1^2+x_2^2-2\rho x_1x_2}{\nu(1-\rho^2)}\right)^{\frac{\nu+2}{2}}\,dx_1\,dx_2$ & \begin{tabular}[c]{@{}l@{}}$\rho\in\{0,0.3,0.5,0.8\}$\\
 $\nu\in\{3,5,10\}$ \end{tabular}\\ \hline
Frank            & $C_\theta(u,v)=-\frac{1}{\theta}\ln\left(1+\frac{\left(e^{-\theta u}-1\right)\left(e^{-\theta v}-1\right)}{e^{-\theta}-1}\right)$ & $\theta\in\{2,3.7,9\}$ \\ \hline 
Gumbel           & $C_\theta(u,v)=\exp\left(-\left[(-\log u)^\theta +(-\log v)^\theta\right]^{1/\theta}\right)$ & $\theta\in\{1.25,1.5,2.5\}$ \\ \hline

Joe              & $C_\theta(u,v)=1-\left[(1-u)^\theta +(1-v)^\theta-(1-u)^\theta(1-v)^\theta\right]^{1/\theta}$ & $\theta\in\{1.4,1.9,4.4\}$\\ \hline
Galambos         & $C_\theta(u,v)=uv\exp\left[(-\log u)^{-\theta}+(-\log v)^{-\theta}\right]^{-1/\theta}$ & $\theta\in\{0.5,0.8,1.8\}$\\ \hline
H\"{u}sler-Reiss & \begin{tabular}[c]{@{}l@{}}$C_\theta(u,v)=\exp\left(-\hat{u}\,\Phi\left[\frac{1}{\theta}+\frac{\theta}{2}\log\frac{\hat{u}}{\hat{v}}\right]-\hat{v}\,\Phi\left[\frac{1}{\theta}+\frac{\theta}{2}\log\frac{\hat{v}}{\hat{u}}\right]\right)$\\ with $\hat{u}=-\log u,\,\hat{v}=-\log v$ \end{tabular}& $\theta\in\{0.85,1.2,2.4\}$ \\ \hline
\end{tabular}
\caption{Summary of copulas used in the power analysis. In the table, by $t^{-1}_\nu$ we denote the quantile function of a standard univariate Student's $t$ distribution with $\nu$ degrees of freedom and $\Gamma$ is the Gamma function. The limiting cases for Gumbel, Joe and Frank copulas are independence copula for $\theta=1$ and complete dependence for $\theta\to\infty$. For Galambos and H\"{u}sler-Reiss copulas the limiting cases are independence copula for $\theta=0$ and complete dependence for $\theta\to\infty$  
.}\label{table.copulas}
\end{table}

First, we consider the Student't t copula with the number of degrees of freedom $\nu\in\{3,5,10\}$ and the (unconditional) correlation $\rho\in\{0,0.3,0.5,0.8\}$; see \cite{demarta2005t} for a detailed discussion. Note that this copula covers a symmetric fat-tailed setting. Second, we analyse the Frank copula with the parameter $\theta\in \{2,3.7,9\}$; this is an exemplary choice for a symmetric light-tailed framework. Next, we study the Gumbel copula with $\theta\in\{1.25,1.5,2.5\}$, the Joe copula with $\theta\in\{1.4,1.9,4.4\}$, the Galambos copula with $\theta\in\{0.5,0.8,1.8\}$ and the H\"{u}sler-Reiss copula with $\theta\in\{0.85,1.2,2.4\}$; these copulas could be used to model asymmetry in the bivariate distributions. For a more detailed discussion on Frank, Gumbel, and Joe copulas we refer to \cite{nelsen2007introduction}, while Galambos and H\"{u}sler-Reiss copulas are covered in Section 4.3 in \cite{genest2009rank}. It should be noted that all the copulas' parameters were chosen in the way that the unconditional correlation of the bivariate random vector with a given copula and standard normal margins would be approximately equal to $\{0.3,0.5,0.8\}$, respectively. For the {\it t}-copula we additionally include $\rho=0$; note that in the {\it t}-copula framework this encodes some form of dependence of the margins while for the other considered models, by setting $\rho=0$ we get the independence (and, consequently, the multivariate normality). Finally, it should be noted that the choice of alternative distributions in this paper is consistent with the ones in~\cite{GenRemBea2009} and~\cite{MalSor2003} to facilitate easier comparison of the results.

The performance estimates are based on the Monte Carlo samples generated in R 4.1.3 with {\it copula} library for copula random number generation. The implementation of the tests BHEP and MS was taken from {\it mnt} library while the remaining benchmark statistics were based on our own code. Since all the considered copulas satisfy the property $C(u,v)=C(v,u)$, $u,v\in (0,1)$, in the testing procedures we use $\alpha=(1,1)$ as a loading factor in $T_n$, $R_n$, $\tilde T_n$, and $\tilde R_n$. This effectively puts equal weight on both margins.  For illustrative purposes, in Figure~\ref{fig:alternatives} we present exemplary samples from the considered copulas with standard normal margins and we indicate the corresponding 20/60/20 split.

\begin{figure}[htp!]
\centering
\includegraphics[width=0.8\textwidth]{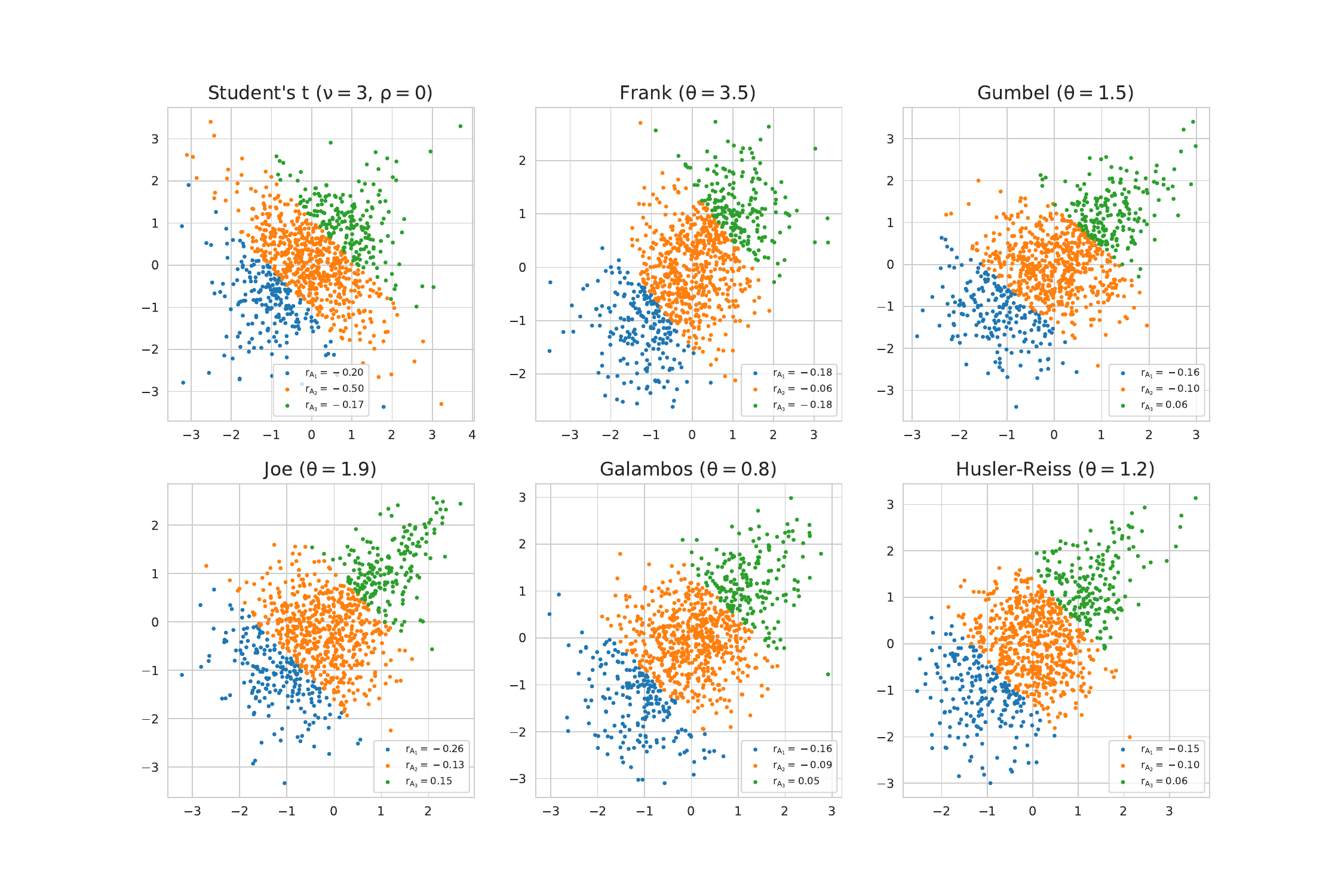}
\caption{The plot presents samples from selected alternative distributions, all with standard normal margins. The benchmark-induced subsamples for equal loading factors and 20/60/20 split, i.e. for sets $A_1$, $A_2$, and $A_3$, are marked in different colours. We denote in legends the conditional covariances for each conditioning set. As expected the quantile conditional covariances on those subsets are not close to each other. See Figure~\ref{F:1} for a similar plot under the null hypothesis.}
\label{fig:alternatives}
\end{figure}


In all cases, we consider four different sample sizes $n\in\{50,100,250,1000\}$. For $n\in\{50,100,250\}$ we used $N=1\,000\,000$ Monte Carlo samples (to evaluate test power) while for $n=1000$, due to time-consuming calculations, we used $N=200\,000$. For each test, first, we simulated the data from the null distribution (normal copula with independent standard normal margins) and estimated rejection regions with the test size $0.05$. Note that for some of the considered test statistics, the null distribution is known, at least in the asymptotic form. Still, we decided to use the simulated null distribution and the corresponding rejection region to obtain better control on the type I error. 

We start with assessing the empirical test size. More specifically, we evaluate the empirical type I error for various choices of the (unconditional) correlation coefficient $\rho$ and check if it is aligned with the chosen test size $0.05$.  In the simulations, we chose $\rho\in \{-0.8,-0.5,-0.3,0.0,0.3,0.5,0.8\}$ and used $\tilde N=200\,000$ (for $n\in\{50,100,250\}$) or $\tilde N=50\,000$ (for $n=1000$) Monte Carlo samples. 
The results are summarised in Table~\ref{T:check0}. For BHEP, AD, CM, and MS directly from the construction, the null distribution for any $n$ does not depend on $\rho$ and the results are very close to $0.05$ (up to the simulation error). For $T_n$ and $R_n$, only the asymptotic ($n\to\infty$) distribution is independent of the null distribution parameters. Still, the empirical type I error is in the reasonable range and converges to $0.05$ as the sample size increases. Also, it should be noted that,  for small sample sizes and tests $T_n$ and $R_n$, the results are asymmetric. More specifically, for $\rho$ close to $-1$, the empirical type one error is slightly below the assumed test size while for $\rho$ close to $1$, the empirical type I error is bigger than the assumed test size. This follows from the fact that, for $\rho\to 1$, the benchmark random variable $Y=\alpha_1X_1+\alpha_2X_2$ converges to $(\alpha_1+\alpha_2)X_1$ while for for $\rho\to -1$, we have $Y\to(\alpha_1-\alpha_2)X_1$, which results in the asymmetry of the test power. Finally, as expected, the empirical type I error for $\tilde T_n$ and $\tilde R_n$ is stable and close to $0.05$. This follows from the fact that these two test statistics are based on the de-correlated samples which makes them effectively invariant with respect to $\rho$.

\begin{table}[h!]
\centering
{\footnotesize
\begin{tabular}{|c c|c c c c |c c |c c|}
    \hline
      Correlation & n & BHEP & AD & CM & MS & $T_n$ & $R_n$ & $\tilde T_n$ & $\tilde R_n$
      \\
    \hline
    \multirow{4}{5em}{$\rho=-0.8$} & 50 & 0.050 & 0.050 & 0.050 & 0.048 & 0.048 & 0.046 & 0.050 & 0.050 \\
       & 100 & 0.049 & 0.050 & 0.050 & 0.050 & 0.048 & 0.048 & 0.050 & 0.050  \\
       & 250 & 0.049 & 0.050 & 0.050 & 0.049 & 0.049 & 0.049 & 0.050 & 0.049\\
       & 1000 & 0.051 & 0.049 & 0.051 & 0.050 & 0.050 & 0.051 & 0.050 & 0.051
       \\
    \hline
    \multirow{4}{5em}{$\rho=-0.5$} & 50 & 0.050 & 0.049 & 0.048 & 0.051 & 0.048 & 0.045 & 0.050 & 0.049 \\
       & 100 & 0.050 & 0.050 & 0.050 & 0.050 & 0.049 & 0.049 & 0.049 & 0.048  \\
       & 250 & 0.049 & 0.050 & 0.050 & 0.050 & 0.050 & 0.050 & 0.050 & 0.049\\
       & 1000 & 0.053 & 0.050 & 0.052 & 0.050 & 0.048 & 0.049 & 0.048 & 0.047
       \\
    \hline
    \multirow{4}{5em}{$\rho=-0.3$} & 50 & 0.050 & 0.050 & 0.050 & 0.050 & 0.049 & 0.047 & 0.050 & 0.049  \\
       & 100 & 0.049 & 0.050 & 0.050 & 0.050 & 0.050 & 0.049 & 0.049 & 0.049  \\
       & 250 & 0.049 & 0.050 & 0.050 & 0.050 & 0.050 & 0.050 & 0.050 & 0.050 \\
       & 1000 & 0.051 & 0.049 & 0.050 & 0.049 & 0.050 & 0.049 & 0.048 & 0.049 
       \\
    \hline
    \multirow{4}{5em}{$\rho=0$} & 50 & 0.050 & 0.050 & 0.050 & 0.050 & 0.050 & 0.050 & 0.050 & 0.050  \\
       & 100 & 0.049 & 0.050 & 0.051 & 0.050 & 0.050 & 0.050 & 0.050 & 0.050  \\
       & 250 & 0.050 & 0.050 & 0.050 & 0.051 & 0.050 & 0.050 & 0.050 & 0.049 \\
       & 1000 & 0.053 & 0.050 & 0.050 & 0.051 & 0.051 & 0.051 & 0.050 & 0.051 
       \\
    \hline
    \multirow{4}{5em}{$\rho=0.3$} & 50 & 0.050 & 0.050 & 0.050 & 0.050 & 0.054 & 0.055 & 0.050 & 0.050  \\
       & 100 & 0.050 & 0.050 & 0.050 & 0.050 & 0.052 & 0.052 & 0.049 & 0.050  \\
       & 250 & 0.049 & 0.050 & 0.050 & 0.049 & 0.051 & 0.051 & 0.050 & 0.049\\
       & 1000 & 0.052 & 0.051 & 0.051 & 0.050 & 0.050 & 0.050 & 0.048 & 0.047 
       \\
    \hline
    \multirow{4}{5em}{$\rho=0.5$} & 50 & 0.050 & 0.050 & 0.050 & 0.050 & 0.059 & 0.060 & 0.050 & 0.050 \\
       & 100 & 0.049 & 0.050 & 0.050 & 0.050 & 0.054 & 0.054 & 0.049 & 0.049  \\
       & 250 & 0.051 & 0.049 & 0.050 & 0.050 & 0.052 & 0.053 & 0.050 & 0.050 \\
       & 1000 & 0.052 & 0.050 & 0.050 & 0.049 & 0.050 & 0.051 & 0.049 & 0.049 
       \\
    \hline
    \multirow{4}{5em}{$\rho=0.8$} & 50 & 0.049 & 0.049 & 0.050 & 0.049 & 0.065 & 0.059 & 0.051 & 0.050 \\
       & 100 & 0.049 & 0.051 & 0.051 & 0.050 & 0.057 & 0.055 & 0.049 & 0.049 \\
       & 250 & 0.050 & 0.050 & 0.050 & 0.050 & 0.054 & 0.052 & 0.050 & 0.050\\
       & 1000 & 0.052 & 0.050 & 0.050 & 0.047 & 0.050 & 0.051 & 0.049 & 0.049 
       \\
    \hline
\end{tabular}}
\caption{Empirical type I error of the underlying test statistics for the samples of size $n$ from the bivariate Gaussian distribution with the correlation coefficient $\rho$ and at the confidence level $5\%$. The rejection threshold was based on $N= 1\,000\,000$ (for $n\in \{50,100,250\}$) or $N = 
200\,000$ (for $n=1000$) Monte Carlo samples from the bivariate normal distribution with independent standard margins.}\label{T:check0}
\end{table}

 Next, we assess the discriminatory power for test statistics $T_n$ and $\tilde{T}_n$ (for symmetric alternatives) and $R_n$ and $\tilde{R}_n$ (for asymmetric alternatives). As before, for every alternative distribution choice, we simulate $\tilde N=200\,000$ (for $n\in\{50,100,250\}$) or $\tilde N=50\,000$ (for $n=1000$) Monte Carlo samples and check the proportion of the samples for which the underlying test rejected the null hypothesis based on the simulated rejection region and the test size $0.05$.  The results for symmetric fat-tail alternatives (Student's t copula) are presented in Table~\ref{T:check1}, for symmetric slim-tail alternatives (Frank copula) in Table~\ref{T:check2}, and for asymmetric alternatives (Gumbel, Joe, Galambos, H\"{u}sler-Reiss copula) in Table~\ref{T:check3}. Note that, in Table~\ref{T:check1}--\ref{T:check2}, we use $T_n$  with both two-sided and one-sided rejection regions; the same applies to $R_n$ in Table~\ref{T:check3}. This is directly related to the interpretation of these statistics. In particular, if the dependence in the tails is expected to be bigger (resp. smaller) than the dependence in the central part, then we expect that the value of $T_n$ is large and it is reasonable to use a right-sided rejection region. However, note that the benchmark test statistics lack a similar interpretation and only a one-sided rejection region can be used; see Remark~\ref{rm:fat_tail} for details. Finally, it should be noted that, for better comparability of the results, when specifying the best performance in Tables~\ref{T:check1}--\ref{T:check3}, we exclude the new procedures with one-sided rejection regions.

\begin{table}[h!]
\centering
{\footnotesize
\begin{tabular}{|c| c c|c c c c c c | c |}
    \hline
      $\nu$ & $\rho$ & n & BHEP & AD & CM & MS & $T_n$ & $\tilde T_n$ & 
           $T_n$ 
           (RS) 
      \\
      \hline
      \multirow{16}{4em}{$\nu=3$} & \multirow{4}{5em}{$\rho=0$}  & 50 & 0.125 & 0.129 & 0.127 & 0.172 & \textbf{0.300} & 0.233 & 0.395  \\
       & & 100 & 0.187 & 0.206 & 0.208 & 0.206 & \textbf{0.492} & 0.449 & 0.594  \\
       & & 250 & 0.413 & 0.429 & 0.438 & 0.239 & 0.827 & \textbf{0.853} & 0.885  \\
       & & 1000 & 0.988 & 0.945 & 0.949 & 0.261 & 0.999 & \textbf{1.000} & 0.999  \\\cline{2-10}
       & \multirow{4}{5em}{$\rho=0.3$} & 50 & 0.137 & 0.133 & 0.131 & 0.188 & \textbf{0.288} & 0.231 & 0.374 \\
       & & 100 & 0.218 & 0.216 & 0.219 & 0.230 & \textbf{0.476} & 0.442 & 0.571   \\
       & & 250 & 0.481 & 0.451 & 0.456 & 0.270 & 0.811 & \textbf{0.845} & 0.869   \\
       & & 1000 & 0.994 & 0.953 & 0.956 & 0.296 & \textbf{0.999} & \textbf{0.999} & 0.999   \\\cline{2-10}
       & \multirow{4}{5em}{$\rho=0.5$} & 50 & 0.165 & 0.144 & 0.142 & 0.216 & \textbf{0.242} & 0.226 & 0.313   \\
       & & 100 & 0.267 & 0.237 & 0.240 & 0.271 & 0.394 & \textbf{0.429} & 0.485  \\
       & & 250 & 0.583 & 0.488 & 0.497 & 0.326 & 0.715 & \textbf{0.828} & 0.792  \\
       & & 1000 & 0.998 & 0.969 & 0.971 & 0.367 & 0.998 & \textbf{1.000} & 0.998   \\\cline{2-10}
       & \multirow{4}{5em}{$\rho=0.8$} & 50 & 0.226 & 0.174 & 0.173 & \textbf{0.274} & 0.124 & 0.214 & 0.158   \\
       & & 100 & 0.384 & 0.295 & 0.297 & 0.361 & 0.159 & \textbf{0.395} & 0.214   \\
       & & 250 & 0.751 & 0.581 & 0.587 & 0.460 & 0.271 & \textbf{0.776} & 0.361   \\
       & & 1000 & 0.999 & 0.988 & 0.988 & 0.554 & 0.711 & \textbf{1.000} & 0.800   \\
    \hline
    \multirow{16}{4em}{$\nu=5$} & \multirow{4}{5em}{$\rho=0$}  & 50 & 0.084 & 0.084 & 0.083 & 0.122 & \textbf{0.152} & 0.120 & 0.217   \\
       & & 100 & 0.102 & 0.114 & 0.113 & 0.143 & \textbf{0.234} & 0.210 & 0.321   \\
       & & 250 & 0.172 & 0.204 & 0.208 & 0.162 & 0.446 & \textbf{0.467} & 0.554   \\
       & & 1000 & 0.630 & 0.611 & 0.624 & 0.178 & 0.926 & \textbf{0.969} & 0.956   \\\cline{2-10}
       & \multirow{4}{5em}{$\rho=0.3$} & 50 & 0.088 & 0.085 & 0.084 & 0.128 & \textbf{0.154} & 0.121 & 0.212  \\
       & & 100 & 0.114 & 0.119 & 0.119 & 0.156 & \textbf{0.234} & 0.206 & 0.314  \\
       & & 250 & 0.199 & 0.211 & 0.217 & 0.181 & 0.435 & \textbf{0.463} & 0.539   \\
       & & 1000 & 0.706 & 0.629 & 0.644 & 0.202 & 0.913 & \textbf{0.966} & 0.948   \\\cline{2-10}
       & \multirow{4}{5em}{$\rho=0.5$} & 50 & 0.099 & 0.090 & 0.089 & \textbf{0.142} & 0.139 & 0.120 & 0.183   \\
       & & 100 & 0.131 & 0.125 & 0.126 & 0.173 & \textbf{0.195} & 0.206 & 0.263   \\
       & & 250 & 0.244 & 0.226 & 0.233 & 0.209 & 0.358 & \textbf{0.453} & 0.457   \\
       & & 1000 & 0.803 & 0.659 & 0.672 & 0.239 & 0.837 & \textbf{0.963} & 0.894   \\\cline{2-10}
       & \multirow{4}{5em}{$\rho=0.8$} & 50 & 0.120 & 0.100 & 0.098 & \textbf{0.169} & 0.091 & 0.118 & 0.109   \\
       & & 100 & 0.176 & 0.144 & 0.145 & \textbf{0.218} & 0.099 & 0.197 & 0.133   \\
       & & 250 & 0.354 & 0.267 & 0.275 & 0.277 & 0.136 & \textbf{0.430} & 0.197   \\
       & & 1000 & 0.931 & 0.739 & 0.749 & 0.348 & 0.331 & \textbf{0.947} & 0.440  \\
    \hline
    \multirow{16}{4em}{$\nu=10$} & \multirow{4}{5em}{$\rho=0$}  & 50 & 0.062 & 0.061 & 0.060 & 0.081 & \textbf{0.083} & 0.069 & 0.114   \\
       & & 100 & 0.067 & 0.069 & 0.068 & 0.092 & \textbf{0.102} & 0.093 & 0.150   \\
       & & 250 & 0.081 & 0.092 & 0.093 & 0.101 & 0.156 & \textbf{0.162} & 0.229  \\
       & & 1000 & 0.167 & 0.205 & 0.216 & 0.104 & 0.405 & \textbf{0.507} & 0.518  \\\cline{2-10}
       & \multirow{4}{5em}{$\rho=0.3$} & 50 & 0.064 & 0.062 & 0.061 & 0.085 & \textbf{0.088} & 0.070 & 0.114   \\
       & & 100 & 0.070 & 0.070 & 0.069 & 0.095 & \textbf{0.105} & 0.091 & 0.146   \\
       & & 250 & 0.086 & 0.093 & 0.095 & 0.105 & 0.158 & \textbf{0.163} & 0.228   \\
       & & 1000 & 0.187 & 0.206 & 0.218 & 0.108 & 0.401 & \textbf{0.499} & 0.514   \\\cline{2-10}
       & \multirow{4}{5em}{$\rho=0.5$} & 50 & 0.065 & 0.063 & 0.061 & 0.087 & \textbf{0.089} & 0.069 & 0.107  \\
       & & 100 & 0.073 & 0.071 & 0.071 & \textbf{0.101} & \textbf{0.101} & 0.091 & 0.135   \\
       & & 250 & 0.093 & 0.094 & 0.096 & 0.113 & 0.141 & \textbf{0.161} & 0.203   \\
       & & 1000 & 0.232 & 0.216 & 0.227 & 0.121 & 0.335 & \textbf{0.496} & 0.443   \\\cline{2-10}
       & \multirow{4}{5em}{$\rho=0.8$} & 50 & 0.070 & 0.064 & 0.063 & \textbf{0.096} & 0.076 & 0.069 & 0.083   \\
       & & 100 & 0.080 & 0.073 & 0.073 & \textbf{0.112} & 0.072 & 0.091 & 0.089  \\
       & & 250 & 0.113 & 0.100 & 0.102 & 0.132 & 0.079 & \textbf{0.159} & 0.112  \\
       & & 1000 & 0.335 & 0.237 & 0.248 & 0.148 & 0.129 & \textbf{0.483} & 0.198  \\
    \hline
\end{tabular}}
\caption{Test power for fat-tail symmetric alternatives (Student's t copula with $\nu$ degrees of freedom and the correlation coefficient $\rho$) at the confidence level $5\%$. The best performance is marked in bold. For completeness, we also added results for $T_n$ (RS), which stands for the test $T_n$ with the right-sided rejection region.
}\label{T:check1}
\end{table}

\begin{table}[h!]
\centering
{\footnotesize
\begin{tabular}{|l c|c c c c c c|c|}
    \hline
      Distribution & n & BHEP & AD & CM & MS & $T_n$ & $\tilde T_n$ & 
           $T_n$ (LS) 
           
      \\
    \hline
    \multirow{4}{*}{    \begin{tabular}{l}
            Frank \\
           $\theta=2$ \\
           $(\rho\approx0.3$)
      \end{tabular}} & 50 & 0.061 & 0.051 & 0.051 & 0.065 & \textbf{0.073} & 0.055 & 0.103  \\
       & 100 & 0.070 & 0.055 & 0.054 & 0.069 & \textbf{0.081} & 0.058 & 0.124  \\
       & 250 & 0.102 & 0.060 & 0.059 & 0.072 & \textbf{0.114} & 0.061 & 0.177  \\
       & 1000 & 0.356 & 0.081 & 0.082 & 0.076 & \textbf{0.280} & 0.076 & 0.389  \\
    \hline
       \multirow{4}{*}{    \begin{tabular}{l}
            Frank \\
           $\theta=3.7$ \\
           $(\rho\approx0.5$)
      \end{tabular}} & 50 & 0.088 & 0.067 & 0.066 & 0.098 & \textbf{0.127} & 0.075 & 0.198  \\
       & 100 & 0.123 & 0.083 & 0.080 & 0.111 & \textbf{0.195} & 0.086 & 0.290  \\
       & 250 & 0.256 & 0.124 & 0.123 & 0.123 & \textbf{0.394} & 0.115 & 0.518  \\
       & 1000 & 0.895 & 0.346 & 0.349 & 0.130 & \textbf{0.912} & 0.243 & 0.952  \\
    \hline
        \multirow{4}{*}{    \begin{tabular}{l}
            Frank \\
           $\theta=9$ \\
           $(\rho\approx0.8$)
      \end{tabular}} & 50 & \textbf{0.247} & 0.231 & 0.222 & 0.242 & \textbf{0.247} & 0.194 & 0.369  \\
       & 100 & \textbf{0.461} & 0.423 & 0.415 & 0.291 & 0.449 & 0.251 & 0.589  \\
       & 250 & \textbf{0.889} & 0.801 & 0.799 & 0.340 & 0.846 & 0.381 & 0.915  \\
       & 1000 & \textbf{1.000} & \textbf{1.000} & \textbf{1.000} & 0.375 & \textbf{1.000} & 0.784 & 1.000  \\
    \hline
\end{tabular}}
\caption{Test power for fat-tail symmetric alternatives (Frank copula) for the confidence level $5\%$. In the table, $\theta$ is the parameter value of the copula and $\rho$ is the corresponding (unconditional) correlation of a sample under the alternative hypothesis. The best performance is marked in bold. For completeness, we also added results for $T_n$ (LS), which stands for the test $T_n$ with the left-sided rejection region.
}\label{T:check2}
\end{table}

\begin{table}[h!]
\centering
{\footnotesize
\begin{tabular}{|l c|c c c c c c|c|}
    \hline
      Distribution & n & BHEP & AD & CM & MS & $R_n$ & $\tilde R_n$ & 
           $R_n$ 
           (RS)   \\
    \hline
        \multirow{4}{*}{    \begin{tabular}{l}
            Gumbel \\
           $\theta=1.25$ \\
           $(\rho\approx0.3$)
      \end{tabular}}  & 50 & 0.065 & 0.056 & 0.055 & 0.099 & \textbf{0.149} & 0.087 & 0.201   \\
       & 100 & 0.077 & 0.063 & 0.062 & 0.147 & \textbf{0.209}& 0.133 & 0.284   \\
       & 250 & 0.121 & 0.074 & 0.075 & 0.306 & \textbf{0.395} & 0.264 & 0.501\\
       & 1000 & 0.477 & 0.136 & 0.142 & 0.857 & \textbf{0.882} & 0.766 & 0.929 \\
    \hline
             \multirow{4}{*}{    \begin{tabular}{l}
            Gumbel \\
           $\theta=1.5$ \\
           $(\rho\approx0.5$)
      \end{tabular}} & 50 & 0.079 & 0.062 & 0.061 & 0.137 & \textbf{0.182} & 0.117 & 0.239   \\
       & 100 & 0.108 & 0.072 & 0.070 & 0.235 & \textbf{0.270} & 0.196 & 0.352   \\
       & 250 & 0.236 & 0.096 & 0.097 & \textbf{0.534} & 0.519 & 0.417 & 0.624  \\
       & 1000 & 0.908 & 0.230 & 0.241 & \textbf{0.992} & 0.972 & 0.940 & 0.986 \\
    \hline
             \multirow{4}{*}{    \begin{tabular}{l}
            Gumbel \\
           $\theta=2.5$ \\
           $(\rho\approx0.8$)
      \end{tabular}} & 50 & 0.162 & 0.089 & 0.088 & \textbf{0.279} & 0.124 & 0171 & 0.162  \\
       & 100 & 0.307 & 0.123 & 0.125 & \textbf{0.516} & 0.150 & 0.293 & 0.202 \\
       & 250 & 0.741 & 0.227 & 0.234 & \textbf{0.910} & 0.245 & 0.593 & 0.327\\
       & 1000 & \textbf{1.000} & 0.679 & 0.692 & \textbf{1.000} & 0.654 & 0.990 & 0.752\\
    \hline
           \multirow{4}{*}{    \begin{tabular}{l}
            Joe \\
           $\theta=1.4$ \\
           $(\rho\approx0.3$)
      \end{tabular}} & 50 & 0.086 & 0.060 & 0.059 & 0.166 & \textbf{0.238} & 0.126 & 0.314 \\
       & 100 & 0.130 & 0.064 & 0.064 & 0.320 & \textbf{0.382} & 0.233 & 0.484  \\
       & 250 & 0.321 & 0.081 & 0.082 & \textbf{0.725} & 0.716 & 0.504 & 0.802 \\
       & 1000 & 0.976 & 0.162 & 0.170 & \textbf{0.999} & \textbf{0.999} & 0.977 & 0.999 \\
    \hline
    \multirow{4}{*}{    \begin{tabular}{l}
            Joe \\
           $\theta=1.9$ \\
           $(\rho\approx0.5$)
      \end{tabular}} & 50 & 0.151 & 0.067 & 0.066 & \textbf{0.336} & 0.291 & 0.220 & 0.373 \\
       & 100 & 0.323 & 0.082 & 0.081 & \textbf{0.673} & 0.478 & 0.389 & 0.582\\
       & 250 & 0.836 & 0.117 & 0.121 & \textbf{0.987} & 0.850 & 0.749 & 0.909 \\
       & 1000 & \textbf{1.000} & 0.315 & 0.331 & \textbf{1.000} & \textbf{1.000} & 0.999 & 1.000  \\
    \hline
    \multirow{4}{*}{    \begin{tabular}{l}
            Joe \\
           $\theta=4.4$ \\
           $(\rho\approx0.8$)
      \end{tabular}} & 50 & 0.698 & 0.250 & 0.253 & \textbf{0.879} & 0.120 & 0.411 & 0.157   \\
       & 100 & 0.975 & 0.443 & 0.453 & \textbf{0.997} & 0.140 & 0.636 & 0.189 \\
       & 250 & \textbf{1.000} & 0.820 & 0.831 & \textbf{1.000} & 0.211 & 0.926 & 0.287 \\
       & 1000 & \textbf{1.000} & \textbf{1.000} & \textbf{1.000} & \textbf{1.000} & 0.554 & \textbf{1.000} & 0.663  \\
    \hline
   \multirow{4}{*}{    \begin{tabular}{l}
            Galambos \\
           $\theta=0.5$ \\
           $(\rho\approx0.3$)
      \end{tabular}} & 50 & 0.063 & 0.054 & 0.053 & 0.091 & \textbf{0.144} & 0.082 & 0.198  \\
    & 100 & 0.072 & 0.055 & 0.054 & 0.141 & \textbf{0.206} & 0.118 & 0.282 \\
       & 250 & 0.103 & 0.059 & 0.059 & 0.306 & \textbf{0.392} & 0.225 & 0.502\\
       & 1000 & 0.390 & 0.076 & 0.077 & 0.885 & \textbf{0.901} & 0.682 & 0.942  \\
    \hline
     \multirow{4}{*}{
     \begin{tabular}{l}
            Galambos \\
           $\theta=0.8$ \\
           $(\rho\approx0.5$)
      \end{tabular}} & 50 & 0.072 & 0.056 & 0.055 & 0.130 & \textbf{0.178} & 0.108 & 0.236  \\
    & 100 & 0.096 & 0.058 & 0.058 & 0.235 & \textbf{0.266} & 0.175 & 0.349 \\
       & 250 & 0.198 & 0.066 & 0.067 & \textbf{0.567} & 0.518 & 0.373 & 0.624 \\
       & 1000 & 0.818 & 0.080 & 0.083 & \textbf{0.996} & 0.977 & 0.907 & 0.989 \\
    \hline
     \multirow{4}{*}{
     \begin{tabular}{l}
            Galambos \\
           $\theta=1.8$ \\
           $(\rho\approx0.8$)
      \end{tabular}} & 50 & 0.145 & 0.075 & 0.075 & \textbf{0.267} & 0.120 & 0.159 & 0.157  \\
    & 100 & 0.270 & 0.098 & 0.099 & \textbf{0.510} & 0.147 & 0.266 & 0.199 \\
       & 250 & 0.681 & 0.158 & 0.164 & \textbf{0.918} & 0.236 & 0.547 & 0.317 \\
       & 1000 & 0.999 & 0.473 & 0.488 & \textbf{1.000} & 0.637 & 0.983 & 0.740  \\
    \hline
         \multirow{4}{*}{
     \begin{tabular}{l}
            H\"{u}sler-Reiss \\
           $\theta=0.85$ \\
           $(\rho\approx0.3$)
      \end{tabular}}
    & 50 & 0.059 & 0.052 & 0.052 & 0.086 & \textbf{0.140} & 0.079 & 0.191 \\
       & 100 & 0.069 & 0.053 & 0.053 & 0.136 & \textbf{0.200} & 0.109 & 0.277 \\
       & 250 & 0.102 & 0.053 & 0.053 & 0.300 & \textbf{0.386} & 0.205 & 0.496  \\
       & 1000 & 0.371 & 0.056 & 0.058 & 0.888 & \textbf{0.899} & 0.626 & 0.941   \\
    \hline
             \multirow{4}{*}{
     \begin{tabular}{l}
            H\"{u}sler-Reiss \\
           $\theta=1.2$ \\
           $(\rho\approx0.5$)
      \end{tabular}}
       & 50 & 0.064 & 0.051 & 0.051 & 0.116 & \textbf{0.174} & 0.099 & 0.231   \\
       & 100 & 0.085 & 0.053 & 0.053 & 0.217 & \textbf{0.258} & 0.155 & 0.343 \\
       & 250 & 0.166 & 0.053 & 0.054 & \textbf{0.555} & 0.514 & 0.319 & 0.625 \\
       & 1000 & 0.782 & 0.055 & 0.056 & \textbf{0.998} & 0.978 & 0.851 & 0.991 \\
    \hline
             \multirow{4}{*}{
     \begin{tabular}{l}
            H\"{u}sler-Reiss \\
           $\theta=2.4$ \\
           $(\rho\approx0.8$)
      \end{tabular}} & 50 & 0.117 & 0.059 & 0.060 & \textbf{0.236} & 0.119 & 0.131 & 0.155  \\
       & 100 & 0.216 & 0.064 & 0.064 & \textbf{0.495} & 0.143 & 0.215 & 0.194 \\
       & 250 & 0.580 & 0.071 & 0.073 & \textbf{0.932} & 0.225 & 0.447 & 0.306 \\
       & 1000 & \textbf{1.000} & 0.118 & 0.121 & \textbf{1.000} & 0.653 & 0.951 & 0.752 \\
    \hline
\end{tabular}}
\caption{Test power for fat-tail asymmetric alternatives for the confidence level $5\%$. In the table, $\theta$ is the parameter value of the copula and $\rho$ is the corresponding (unconditional) correlation of a sample under the alternative hypothesis. The best performance is marked in bold. For completeness, we also added results for $R_n$ (RS), which stands for the test $R_n$ with the right-sided rejection region.}\label{T:check3}
\end{table}

From Tables~\ref{T:check1}--\ref{T:check3} one can see that the proposed test statistics outperform the existing methodologies in the majority of the considered cases. In particular, in the fat-tail Student's t copula setting (Table~\ref{T:check1}), we can see that $T_n$ or $\tilde{T}_n$ give the best results in most of the considered parameter choices. In this framework, as expected, the power of all tests decreases with respect to $\nu$ (with fixed $\rho$ and $n$) as Student's t distribution with $\nu\to \infty$ converges to the Gaussian case. When we fix $\nu$ and $n$, it can be seen that the power of benchmark statistics is increasing with respect to $\rho$ while the power of $T_n$ decreases in $\rho$. However, for $\nu=3$, $\rho=0.3$, and $n=100$, the power of $T_n$ is more than twice as large as the power of benchmark statistics. This shows that the proposed methodology is able to detect non-normality even in the close-to-uncorrelated case (note that, however, here this does not correspond to the independence of margins), where the other tests fail. For high correlation, instead of $T_n$ it is advisable to use $\tilde T_n$ as its power is close to or even better than the power of the best benchmark (usually, BHEP or MS). Also, it should be noted that, for fixed $\nu$ and $n$, the power of $\tilde T_n$ is stable with respect to $\rho$. This follows directly from the fact that this test statistic uses de-correlated data. Finally, it should be noted that the test power is even better if we consider $T_n$  with right-sided rejection regions.

 For the slim-tailed Frank copula setting (Table~\ref{T:check2}), the best power is visible for BHEP and $T_n$. Again, the performance could be improved by using the test with left-sided rejection region. This follows from the fact that for slim-tailed alternative, we expect the covariance in the left tail to be smaller than the covariance in the central part, which results in low values of $T_n$. 
 
The results for asymmetric alternatives (Table~\ref{T:check3}) confirm the good performance of the proposed methodology, with the strongest competitors being BHEP and MS. In particular, for small correlation $T_n$ and $\tilde T_n$ gives the best power while for high $\rho$, the best performance is exhibited by MS.

 It should be noted that the results presented in Tables~\ref{T:check1}--\ref{T:check3} are consistent with the results from~\cite{MalSor2003, AmeSen2020,GenRemBea2009}, although there are some differences in the setting, parameter choice, and the included tests. In particular, comparing Table 1 in~\cite{MalSor2003} with the results in Table~\ref{T:check1} with $n=1000$ in the present paper, we can see that the proposed tests give similar or better results than the statistics considered therein. Similarly, comparing Table 1 in~\cite{GenRemBea2009} with the results for Frank (2) in Table~\ref{T:check2} for $n=100$, we can see that for a similar dependence parameter value, power of $T_n$ is close to $0.08$, while any test from~\cite{GenRemBea2009} cannot reach that level, even for bigger sample size $n=150$.
 
To better explain the discriminatory power results, in Figure~\ref{F:why.good} we show the empirical distributions of $T_n$ and BHEP under the null (Gaussian) and selected alternative (t-Student with 3 degrees of freedom and $\rho=0.3$) distribution for $n=250$. From the figure, one can see that the null and alternative distribution of BHEP significantly overlap while the distributions of null and alternative distribution of $T_n$ are distinctive which results in higher test power.

\begin{figure}[ht]
\centering
\includegraphics[width=0.5\textwidth]{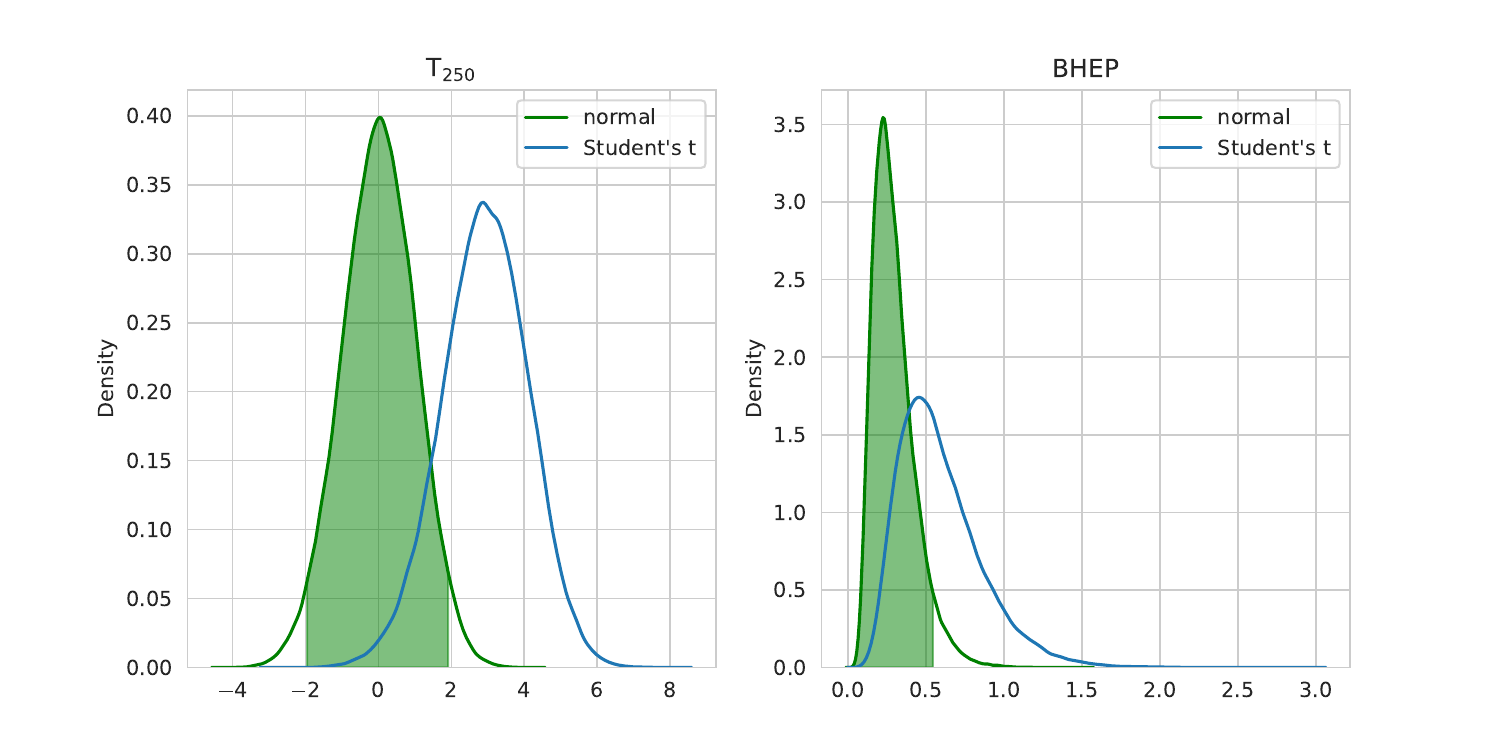}
\caption{Distributions of $T_n$ and $\textrm{BHEP}$ for $n=250$ under the normal copula compared with an exemplary alternative distribution, namely Student's t copula ($\nu=3$, $\rho=0.3$) with normal margins. The results are based on a Monte Carlo sample with N=100,000 elements. The shaded regions correspond to 95\% confidence interval for the null distribution.} 
\label{F:why.good}
\end{figure}


\section{Empirical example: metal commodities}\label{S:empirical}

In this section, we present how to apply our framework in practice, by introducing a simple empirical study. Following ~\cite{MalSor2003}, we decided to consider daily returns for selected metal commodities listed in the London Metal Exchange (LME). We check if our framework can detect the non-normality of the dependence structure in the data and how it compares to other benchmark methods. More specifically, we consider daily spot mid-prices for six metals commodities: aluminium (Al), copper (Cu), lead (Pb), nickel (Ni), tin (Sn) and zinc (Zn). The sample contains 2528 data points from 4.07.2010 to 31.12.2019. This choice of commodities is consistent with~\cite{MalSor2003}, where the data for the same metals from the time interval 1989--1997 were analysed and proven to often exhibit a non-normal behaviour.

For each metal, we compute logarithmic daily returns and, using the empirical CDF, we transform the marginals to the standard normal setting.  In other words, the sample data is transformed to have standard normal margins so that we can focus on the dependence structure check. For completeness, in Figure~\ref{F:empirical} we present the unconditional correlation matrix as well as an exemplary scatterplot for a pair of transformed returns.

\begin{figure}[htp!]
\centering
\begin{subfigure}{0.45\textwidth}  
\centering
\scalebox{0.7}{
\begin{tabular}{|c|c c c c c c|
}
\hline
Metal & Al   & Cu   & Pb   & Ni   & Sn   & Zn   \\
\hline
Al    & 1.00    & 0.58 & 0.51 & 0.49 & 0.39 & 0.56 \\
Cu    & 0.58 & 1.00    & 0.62 & 0.59 & 0.48 & 0.69 \\
Pb    & 0.51 & 0.62 & 1.00    & 0.48 & 0.43 & 0.71 \\
Ni    & 0.49 & 0.59 & 0.48 & 1.00    & 0.43 & 0.44 \\
Sn    & 0.39 & 0.48 & 0.43 & 0.43 & 1.00    & 0.44 \\
Zn    & 0.56 & 0.69 & 0.71 & 0.44 & 0.44 & 1.00 \\
\hline
\end{tabular}}
\end{subfigure}
\begin{subfigure}{0.45\textwidth}  
\centering
\includegraphics[width=0.5\textwidth]{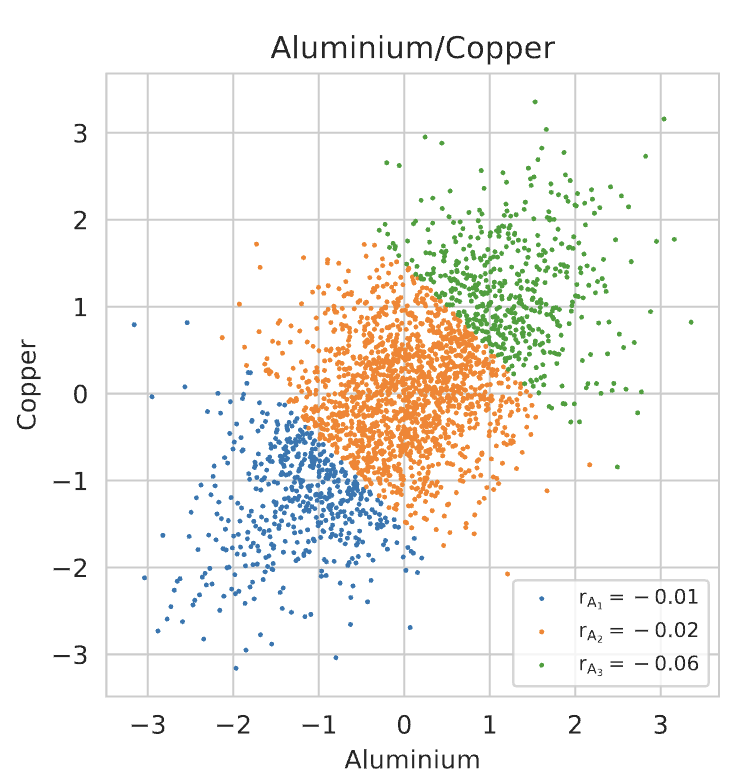}
\end{subfigure}
\caption{The plot presents the unconditional correlation matrix (left) as well as the scatterplot for a selected pair of metal commodities (right). Sample data is transformed to have standard normal margins so that we focus on the dependence structure check. The benchmark-induced subsamples for the loading factor $\alpha=(1,1)$, i.e. for the sets $A_1$, $A_2$, and $A_3$ in the right exhibit are marked in different colours.}
\label{F:empirical}
\end{figure}

Before performing a pairwise statistical test, let us investigate the structure of sample conditional correlations. Namely, we consider a multivariate equal-proportion benchmark applied to the entire data set. We take the loading factor $\alpha=(1,1,1,1,1,1)$, split the whole sample into three subsets according to benchmark values, and then compute the 20/60/20 rule-induced conditional correlation matrices. Potential differences in the conditional correlation matrices indicate sample non-normality, see Remark~\ref{rem:multidimensional}. The results are presented in Table~\ref{T:metal_cor1}.

\begin{table}[h!]
\centering
\scalebox{0.65}{
\begin{tabular}{|c|c c c c c c|}
\hline
Metal & Al    & Cu    & Pb    & Ni    & Sn    & Zn    \\
\hline
Al    & 1.00 & 0.25 & 0.21 & 0.10 & 0.09 & 0.19 \\
Cu    & 0.25 & 1.00 & 0.29 & 0.15 & 0.23 & 0.40 \\
Pb    & 0.21 & 0.29 & 1.00 & 0.11 & 0.13 & 0.50 \\
Ni    & 0.10 & 0.15 & 0.11 & 1.00 & 0.05 & 0.16 \\
Sn    & 0.09 & 0.23 & 0.13 & 0.05 & 1.00 & 0.08 \\
Zn    & 0.19 & 0.40 & 0.50 & 0.16 & 0.08 & 1.00 \\
\hline
\end{tabular}
\,\,\,\,\,\,
\begin{tabular}{|c|c c c c c c|}
\hline
Metal & Al    & Cu    & Pb    & Ni    & Sn    & Zn    \\
\hline
Al    & 1.00 & 0.17 & 0.07  & 0.06  & 0.00  & 0.11  \\
Cu    & 0.17 & 1.00 & 0.18  & 0.18  & 0.02  & 0.27  \\
Pb    & 0.07 & 0.18 & 1.00  & -0.01 & -0.03 & 0.37  \\
Ni    & 0.06 & 0.18 & -0.01 & 1.00  & 0.05  & 0.08  \\
Sn    & 0.00 & 0.02 & -0.03 & 0.05  & 1.00  & -0.03 \\
Zn    & 0.11 & 0.27 & 0.37  & 0.08  & -0.03 & 1.00  \\
\hline
\end{tabular}
\,\,\,\,\,\,
\begin{tabular}{|c|c c c c c c|}
\hline
Metal & Al    & Cu    & Pb    & Ni    & Sn    & Zn    \\
\hline
Al    & 1.00  & 0.12 & 0.02 & 0.09  & -0.05 & 0.13  \\
Cu    & 0.12  & 1.00 & 0.23 & 0.18  & 0.09  & 0.26  \\
Pb    & 0.02  & 0.23 & 1.00 & 0.04  & 0.10  & 0.40  \\
Ni    & 0.09  & 0.18 & 0.04 & 1.00  & 0.00  & -0.02 \\
Sn    & -0.05 & 0.09 & 0.10 & 0.00  & 1.00  & -0.03 \\
Zn    & 0.13  & 0.26 & 0.40 & -0.02 & -0.03 & 1.00  \\
\hline
\end{tabular}
}
\caption{Empirical conditional correlation matrix on $A_1$ (left), $A_2$ (center), and $A_3$ (right) of the  ''Gaussianised'' logarithmic daily returns for loading factor $\alpha=(1,1,1,1,1,1)$. For many pairs, the conditional correlation within the set $A_1$ is larger compared to the set $A_2$ or $A_3$, suggesting an increased dependence in the area of 20\% worst benchmark outcomes.}\label{T:metal_cor1}
\end{table}

From Table~\ref{T:metal_cor1} we infer that the dependence structure increases in the $A_1$ set (bottom 20\% benchmark subsample), at least compared to sets $A_2$ and $A_3$. This suggests that the usage of test statistics $L_n$ (or $\tilde L_n$) should give the best results and further pairwise test is advised. 

Next, we revert to the pairwise setup and perform the statistical testing procedure introduced in this paper. Based on the transformed data, we test if the bivariate copula for each pair of assets agrees with the Gaussian dependence structure. To do this, we apply the proposed tests $T_n,L_n$, $R_n$, $\tilde T_n,\tilde L_n$, and $\tilde R_n$ with loading factor $\alpha=(1,1)$  along with all benchmark tests introduced in Section~\ref{SS:benchmark}, i.e. BHEP, AD, CM, and MS. As in Section~\ref{SS:power}, first we simulate the distribution of the tests statistics under the null hypothesis based on $N=20\,000$ Monte Carlo samples with size $n=2527$. Then we estimate the p-value corresponding to each bivariate sample.

\begin{table}[h!]
\centering
\footnotesize{
\begin{tabular}{|c c| c c c c |c c c |c c c|}
\hline
   \multicolumn{2}{|c|}{Pair}   & BHEP           & AD         & CM    & MS    & $T_n$ & $L_n$      & $R_n$ & $\tilde T_n$   & $\tilde L_n$   & $\tilde R_n$ \\
   \hline
Al & Cu & 0.204          & 0.352      & 0.303 & 0.193 & 0.467 & 0.536      & 0.074 &  {0.073} &  {0.008}          &  {0.881}        \\
Al & Pb &  {0.034}          & 0.052      & 0.087 & 0.254 & 0.832 & 0.165      & 0.294 &  {0.042}          &  {0.026}     & 0.335        \\
Al & Ni & 0.301          & 0.641      & 0.539 & 0.180  & 0.506 & 0.132      & 0.654 & 0.108          &  {0.026} & 0.893        \\
Al & Sn & 0.740         & 0.660       & 0.539 & 0.085 & 0.550  &  {0.032}      & 0.229 &  0.121          &  {0.012} & 0.941        \\
Al & Zn &  {0.007}          &  {0.018}      &  {0.039} & 0.954 & 0.595 & 0.305      & 0.845 &  {0.004}     &  {0.021}     &  {0.026}   \\
Cu & Pb & 0.051          & 0.280       & 0.203 & 0.404 & 0.995 & 0.330       & 0.324 &  {0.001}     &  {0.001}     &  0.066        \\
Cu & Ni &  {0.001} &  {0.015}      &  {0.013} & 0.080  & 0.262 & 0.772      &  {0.039} & {0.000}          &  {0.000}          & {0.012}        \\
Cu & Sn &  {0.028}          &  {0.014}      &  {0.014} &  {0.004} &  {0.006} &  {0.000} & 0.827 &  {0.000}     &  {0.000}     &  0.094        \\
Cu & Zn &  {0.045}          & 0.452      & 0.424 &  {0.024} & 0.967 & 0.305      & 0.333 &  0.084          &  0.058     & 0.405        \\
Pb & Ni &  {0.032} & 0.313      & 0.208 & 0.077 & 0.340  &  {0.027}      & 0.470  & {0.001}          & {0.000}          & {0.004}        \\
Pb & Sn & 0.089         & 0.098      & 0.082 &  {0.017} &  {0.016} &  {0.002}      & 0.468 &  {0.001}     &  {0.000}     &  {0.038}   \\
Pb & Zn &  {0.000}     &  {0.001}      &  {0.002} & 0.474 &  {0.026} &  {0.016}      & 0.290  &  {0.000}     &  {0.004}     &  {0.001}        \\
Ni & Sn & 0.327          & 0.470       & 0.348 & 0.103 & 0.467 &  {0.033}      & 0.328 &  0.278          &  0.053 & 0.821        \\
Ni & Zn &  {0.000}     &  {0.000} &  {0.004} &  {0.001} & 0.407 & 0.217      &  {0.009} & 0.133          &  0.121 & 0.415        \\
Sn & Zn & 0.532         & 0.326      & 0.417 &  {0.002} &  {0.01}  &  {0.001}      & 0.380  &  {0.012}          &  {0.038}     & 0.066  \\ 
\hline
\multicolumn{2}{|c|}{\% rejections} & 53\% & 33\% & 33\% & 33\% & 27\% & 47\% & 13\% & 60\% & 80\% & 33\%\\
\hline
\end{tabular}
}
\caption{P-values of the corresponding test computed for each pair of metals. The last row, for each test, indicates the percentage of pairs (out of 15) when the null hypothesis was rejected at the level $0.05$. In most cases, test statistics $\hat T_n$ and $\hat L_n$ lead to materially smaller $p$-values.}\label{T:empirical}
\end{table}

The results are reported in Table~\ref{T:empirical}. It can be seen that for each pair, the Gaussianity of the dependence structure was rejected by at least one test at the level $0.05$. In fact, even if we account for multiple comparison effects, the Gaussian copula structure is rejected for at least 9 out of 15 pairs (note that the significance level adjusted by the most conservative Bonferroni correction rule should equal $0.005$). The biggest numbers of rejections were reported by $\tilde L_n$ (12), $\tilde T_n$ (9), BHEP (8), and $L_n$ (7). The relatively high rejection rate for $L_n$ could be associated with the observation that the negative returns tend to occur simultaneously on both coordinates. This increases the dependence in the left tail which results in a large value of $L_n$ and rejection of the null hypothesis. The rejection rate for $\tilde T_n$ and $\tilde L_n$ are even bigger which could be linked to the fact that the empirical (unconditional) correlation for the considered data is relatively high; see Figure~\ref{F:empirical}. This corresponds to the observations from Section~\ref{SS:power}, where we noted that often $\tilde T_n$ outperforms the other statistics when the data is highly correlated. Finally, it should be noted that the results presented in Table~\ref{T:empirical} are consistent with the findings from Section 4.2 in~\cite{MalSor2003}. Indeed, using the same set of commodities in the period 1989--1997, the authors noticed that the Gaussian copula is not the right choice for the modelling of the dependence between metals' returns. We also note that in most cases our testing framework gives much smaller $p$-values, when confronted with benchmark alternatives, which is consistent with the power study results presented in Section~\ref{SS:power}.


\appendix

\section{Mathematical framework and proofs}\label{appendix.A}

In this section, we provide a series of technical lemmas and results which constitute the proof of Theorem~\ref{theorem}. Also, we comment on the explicit values of the constants $K_1,K_2,K_3$ used in \eqref{c.tau}. To do this, we introduce a series of auxiliary results that could be of separate interest. In particular, we show the consistency and the asymptotic normality of the conditional covariance estimators in the bivariate normal setting.

First, let us recall some definitions and notations from Section~\ref{SS:Statistic}. Let $X=(X_i)_{i=1}^{\infty}$ be a sequence of two-dimensional random (sample) vectors, where $X_i=(X_{i,1},X_{i,2})$ are i.i.d. with normal distribution and the parameters given by
\begin{align}\label{eq:X_params_init}
\mathbb{E}(X_{i})=(\mu_1,\mu_2) \quad \sigma(X_{i})=(\sigma_1,\sigma_2), \quad \cov(X_{i,1},X_{i,2})=r,\quad \rho:=\frac{r}{\sigma_1\sigma_2} \in (-1,1),
\end{align}
for some $\mu_1, \mu_2, r\in\mathbb{R}$, $\sigma_1, \sigma_2>0$.
For a fixed vector $\alpha=(\alpha_1,\alpha_2)\in\mathbb{R}^2$, we consider a sequence of (sample) benchmarks $Y=(Y_i)_{i=1}^\infty$ given by 
\begin{align}\label{eq:Y}
    Y_i:= \alpha_1X_{i,1}+\alpha_2X_{i,2}
\end{align}
and a sequence $U=(U_i)_{i=1}^\infty$ given by $
    U_i:=\cov(Y_1,X_{1,2})X_{i,1}-\cov(Y_1,X_{1,1})X_{i,2}$.
By direct computation, for any $i\in \mathbb{N}\setminus \{0\}$, we have $\cov(Y_i,U_i)=0$ and consequently $Y_i$ and $U_i$ are independent (since they follow the bivariate normal distribution).
Now, let $V_i$ and $W_i$ denote the standardized versions of $Y_i$ and $U_i$ given by
\begin{align}\label{eq:V_i}
   V_i:=\frac{Y_i-\bE(Y_i)}{\sigma(Y_i)} \quad \text{and} \quad
   W_i:=\frac{U_i-\bE(U_i)}{\sigma(U_i)}.
\end{align}
Thus, for $j=1,2$, there exist $\beta_{1,j},\beta_{2,j}\in\bR$ such that, for $i=1,2,\ldots$, we have
    $\frac{X_{i,j}-\mu_j}{\sigma_j}=\beta_{1,j} V_i + \beta_{2,j} W_i.$
Since the variance of $\frac{X_{i,j}-\mu_j}{\sigma_j}$ is equal to one, we have $\beta_{1,j}^2+\beta_{2,j}^2=1$ for $j=1,2$. 
Thus, there exist $\gamma_1,\,\gamma_2\in[-\pi,\pi)$, such that $\beta_{1,j}=\cos(\gamma_j)$ and $\beta_{2,j}=\sin(\gamma_j)$, $j=1,2$. Hence, we obtain
\begin{align}\label{eq:X.repr}
    X_{i,j}=\mu_j+\sigma_j\left(\cos(\gamma_j)V_i+\sin(\gamma_j)W_i\right).
\end{align}
Furthermore, noting that $\cos(\gamma_j)=\cov\left(V_1,\frac{X_{1,j}-\mu_j}{\sigma_j}\right)$ and using~\eqref{eq:X.repr} we have 
\begin{align}\label{eq:X_params}
\cos(\gamma_j)=\frac{\cov(Y_1,X_{1,j})}{\sigma(Y_1) \sigma_j} \quad \text{and} \quad
\rho = \cos(\gamma_1 - \gamma_2).
\end{align}

Now, we fix $n\in\mathbb{N}\setminus \{0\}$ and examine the ordering of $X$ and $W$ based on the order statistics of $Y$. As in Section~\ref{S:application2}, by $Y_{(i)}$ we denote the $i$-th smallest  element of the sample $\left(Y_1,\ldots,Y_n\right)$. Similarly, $V_{(i)}$ is the $i$-th smallest of $\left(V_1,\ldots,V_n\right)$. By $(X_{(1)},\ldots, X_{(n)})$ we denote the sequence based on $(X_1,..,X_n)$ with the order given by the ordering of $(Y_1, \ldots, Y_n)$. Furthermore, we index $W_{(i:n)}$ so that it corresponds to the ordering of $Y_{(i)}$. More formally, we consider a family of random variables $G_i$ such that $Y_{(i)}(\omega)=Y_{G_i(\omega)}(\omega) $ and we set 
\begin{align}\label{eq.order}
    X_{(i)}(\omega):=X_{G_i(\omega)}(\omega) \quad \text{and}\quad
    W_{(i:n)}(\omega):=W_{G_i(\omega)}(\omega).
\end{align}
Note that, since $V_i$ and $Y_i$ are comonotonic, we have $V_{G_i(\omega)}(\omega)=V_{(i)}(\omega)$. Furthermore, the sequence $(W_i)_{i=1}^\infty$ is independent of the process $(Y_i)_{i=0}^\infty$, so $W_{(i:n)}$ are i.i.d. with distribution $N(0,1)$ and $(W_{(i:n)})_{i=1}^n$ is independent of $(V_{i})_{i=1}^n$. 

Now, recalling \eqref{eq:X.repr}, for any $i=1,\ldots, n$ and $j=1,2$, we obtain 
\begin{align*}
X_{(i),j} &= \mu_j+ \sigma_j ( \cos(\gamma_j) V_{(i)} + \sin(\gamma_j) W_{(i:n)}),\\ 
(X_{(i),j}-\mu_j)^2 &=  \sigma_j^2 ( \cos^2(\gamma_j) V_{(i)}^2 + \sin(2\gamma_j ) V_{(i)}W_{(i:n)} +\sin^2(\gamma_j) W_{(i:n)}^2),\\
(X_{(i),1}-\mu_1)(X_{(i),2}-\mu_2) &=  \sigma_1 \sigma_2 \left( \cos(\gamma_1) \cos(\gamma_2)V_{(i)}^2
+\sin(\gamma_1+\gamma_2) V_{(i)} W_{(i:n)}
 + \sin(\gamma_1)\sin(\gamma_2) W_{(i:n)}^2 \right).
\end{align*}
Plugging these into~\eqref{eq:cond.estimators:mu} and~\eqref{eq:cond.estimators}, for $j=1,2$, we obtain
\begin{align}
\hat{\mu}_{A,j} &= \mu_j+ \frac{\sigma_j}{m_n} \left( \cos(\gamma_j) \sum_{i=[na]+1}^{[nb]} V_{(i)} 
+ \sin(\gamma_j) \sum_{i=[na]+1}^{[nb]} W_{(i:n)} \right),\label{eq:mu_A:hat}\\
\hat{r}_{A} &= \frac{\sigma_1 \sigma_2}{m_n}
  \left(\cos(\gamma_1) \cos(\gamma_2)  \sum_{i=[na]+1}^{[nb]}V_{(i)}^2 
+\sin(\gamma_1+\gamma_2)  \sum_{i=[na]+1}^{[nb]}V_{(i)} W_{(i:n)}+
  \sin(\gamma_1)\sin(\gamma_2)\sum_{i=[na]+1}^{[nb]}W_{(i:n)}^2 \right)\nonumber\\
&\phantom{=}- (\hat{\mu}_{A,1}-\mu_1)(\hat{\mu}_{A,2}-\mu_2 ).\label{eq:r_A:hat}
\end{align}
In the following, to ease the notation, for any $0\leq a<b\leq1$, we define
\begin{align}
A_V&=A_V(a,b):=\left\{v\in\bR\colon \Phi^{-1}(a)<v<\Phi^{-1}(b)\right\},\\
\lambda_{A,1} &:= \mathbb{E}[V_1\mid V_1\in A_V]=\frac{1}{b-a}\int_a^b \Phi^{-1}(p) \, dp = \frac{1}{b-a}(\phi(\Phi^{-1}(a))-\phi(\Phi^{-1}(b))),\label{eq:lambda^1}\\
\lambda_{A,2} &:=\mathbb{E}[V_1^2\mid V_1\in A_V]=\frac{1}{b-a}\int_a^b (\Phi^{-1}(p))^2 \, dp = 1+\frac{1}{b-a}(\Phi^{-1}(a)\phi(\Phi^{-1}(a))-\Phi^{-1}(b)\phi(\Phi^{-1}(b)));\label{eq:lambda^2}
\end{align}
note that the right hand sides of~\eqref{eq:lambda^1} and~\eqref{eq:lambda^2} follow from the straightforward computation, see  Section 13.10.1 in~\cite{JohKotBal1994} for details.
Next, note that due to~\eqref{eq:V_i}, we have
\begin{align*}
    \mathbb{E}[V_1\mid V_1\in A_V] = \mathbb{E}[V_1\mid Y_1\in A_Y], \qquad\mathbb{E}[V^2_1\mid V_1\in A_V] = \mathbb{E}[V^2_1\mid Y_1\in A_Y],
\end{align*}
where we $A_Y=\left\{y\in\bR \colon F^{-1}_{Y_1}(a)<y<F^{-1}_{Y_1}(b)\right\}$.
Using this notation, later in Proposition~\ref{l.mur}, we will show that the sample moments given in~\eqref{eq:mu_A:hat} and~\eqref{eq:r_A:hat} are consistent and asymptotically normal estimators of 
\begin{align}
    \mu_{A,j}:&=\bE\left[X_{1,j}\mid Y_1\in A\right]= \mu_j+\sigma_j\cos(\gamma_j)\lambda_{A,1},\label{eq:mu_A}\\
    r_A:&=\cov\left[X_{1,1},X_{1,2}\mid Y_1\in A\right] =\sigma_1 \sigma_2 \left( \cos(\gamma_1) \cos(\gamma_2) (\lambda_{A,2}- (\lambda_{A,1})^2) + \sin(\gamma_1)\sin(\gamma_2) \right),\label{eq:r_A}
\end{align}
where, to find the expectations, we used~\eqref{eq:X.repr} and the fact that $V_i$ and $W_i$ are independent.
 
 We start the argument with a technical lemma which will be used multiple times in the proofs. In the lemma, for any sequence $(z_i)$ and $k,l\in \mathbb{N}$, we use the notation of a directed sum given by
    \begin{align}\label{eq:directed_sum}
        \mathcal{E}_{i=k}^l\,z_i:=\begin{cases} 
      \sum_{i=k+1}^l\,z_i, & \text{if } k<l, \\
      0, & \text{if } k=l, \\
      -\sum_{i=l+1}^k\,z_i, & \text{if } k>l;
   \end{cases}
    \end{align}
see~\cite{Stigler1973} and Section 4 in~\cite{JelPit2021} for a similar setup and more details.

\begin{lemma}\label{l:binom}
    Let $(I_k)$ be an i.i.d sequence of random variables following the Bernoulli distribution with the probability of success $a\in (0,1)$. For $n\in \mathbb{N}\setminus \{0\}$, define $A_n:=I_1+\ldots+I_n$. Let $(\tilde{W}_i)$ be a sequence of i.i.d. random variables, independent of $(I_k)$ with $\bE[\tilde{W}_i]=0$ and $\bE\left[\tilde{W}_i^2\right]=S^2<\infty$. Then, we have
    \begin{align*}
        \frac{1}{\sqrt{n}}\mathcal{E}_{i=[na]}^{A_n}\tilde{W}_i\xrightarrow{\bP}0,\qquad n\to\infty.
    \end{align*}
\end{lemma}
\begin{proof}
    We show that the sequence converges to $0$ in $L^2$, which, by Chebyshev's inequality, implies the convergence in probability. Using the fact that $(W_i)$ are independent of $(I_k)$, we have
    \begin{align}\label{eq:l:binom:2}
        \bE\left[\left(\frac{1}{\sqrt{n}}\mathcal{E}_{i=[na]}^{A_n}\tilde{W}_i\right)^2\right]&=\frac{1}{n}\sum_{k=0}^n\bE\left[\left(\mathcal{E}_{i=[na]}^{k}\tilde{W}_i\right)^21_{\{A_n=k\}}\right]\nonumber=\frac{1}{n}\sum_{k=0}^n\bE\left|\mathcal{E}_{i=[na]}^{k}\tilde{W}_i^2\right|\bP[A_n=k]\nonumber=\frac{S^2}{n}\sum_{k=0}^n\big|k-[na]\big|\,\bP[A_n=k]\nonumber\\
        &=S^2\left(\frac{[na]}{n}\left(\sum_{k=0}^{[na]}\bP[A_n=k]-\sum_{k=[na]+1}^n\bP[A_n=k]\right)+\sum_{k=[na]+1}^n\frac{k}{n}\,\bP[A_n=k]-\sum_{k=0}^{[na]}\frac{k}{n}\,\bP[A_n=k]\right).
    \end{align}
    For $n\in\mathbb{N}\setminus \{0\}$, let $N_n$ be a random variable with normal distribution $N\left(a,\sqrt{\frac{a(1-a)}{n}}\right)$. By the Central Limit Theorem, the distributions of $\frac{A_n}{n}$ and $N_n$ are asymptotically equal, i.e., for any $t\in \mathbb{R}$, we have $\lim_{n\to\infty} \left(\bP[\frac{A_n}{n}\leq t]-\bP[{N_n}\leq t]\right)=0$. Thus, we get
    \begin{align}\label{eq:l:binom:3}
        \lim_{n\to\infty}\left(\frac{[na]}{n}\left(\sum_{k=0}^{[na]}\bP[A_n=k]-\sum_{k=[na]+1}^n\bP[A_n=k]\right)\right)&=\lim_{n\to\infty}\frac{[na]}{n}\left(\bP[A_n\leq na]-\bP[A_n> na]\right)\nonumber\\
        &=a\lim_{n\to\infty}\left(\bP[N_n\leq a]-\bP[N_n> a]\right)= a(1/2-1/2)=0.
    \end{align}
    Next, we have
    \begin{align}\label{eq:l:binom:4}
        \lim_{n\to\infty}\left(\sum_{k=[na]+1}^n\frac{k}{n}\bP[A_n=k]-\sum_{k=0}^{[na]}\frac{k}{n}\bP[A_n=k]\right)=\lim_{n\to\infty}\left(\bE\left[\frac{A_n}{n}1_{\{A_n> na\}}\right]-\bE\left[\frac{A_n}{n}1_{\{A_n\leq  na\}}\right]\right).
    \end{align}
    Also, using the law of large numbers and the fact that $\left|\frac{A_n}{n}\right|\leq 1$, we get $\bE\left[\left|\frac{A_n}{n}-a\right|\right]\to 0$ as $n\to\infty$. Thus, noting that $\bE\left[\left|\frac{A_n}{n}-a\right|1_{\{A_n\leq na\}}\right]\leq\bE\left[\left|\frac{A_n}{n}-a\right|\right]$,
    we have 
$\bE\left[\left|\frac{A_n}{n}-a\right|1_{\{A_n\leq na\}}\right]\xrightarrow{n\to\infty} 0$, which implies
    \begin{align}\label{eq:l:binom:5}
        \lim_{n\to\infty}\bE\left[\frac{A_n}{n}1_{\{A_n\leq na\}}\right]=\lim_{n\to\infty}\bE\left[a1_{\{A_n\leq na\}}\right]=a\lim_{n\to\infty}\bP\left[\frac{A_n}{n}\leq a\right]=\frac{a}{2}.
    \end{align}
    In the same way, we can show that $  \bE\left[\frac{A_n}{n}1_{\{A_n> na\}}\right]\xrightarrow{n\to\infty}\frac{a}{2}$. 
    This combined with~\eqref{eq:l:binom:5} shows that \eqref{eq:l:binom:4} converges to $0$. Hence, recalling  \eqref{eq:l:binom:3}, we get that the right-hand side of \eqref{eq:l:binom:2}  converges to $0$, which concludes the proof.
\end{proof}

Now, we present a lemma related to the alternative representation of certain trimmed sums. To this end, for brevity, for any $i\in \mathbb{N}\setminus \{0\}$, we define the supplementary random variables
\begin{align}
        Z^A_{i,1}&:=(V_i-\lambda_{A,1})1_{\{V_i\in A_V\}}+(1_{\{V_i\leq \Phi^{-1}(a)\}}-a)(\Phi^{-1}(a)-\lambda_{A,1})+(b-1_{\{V_i\leq \Phi^{-1}(b)\}})(\lambda_{A,1}-\Phi^{-1}(b)),\label{eq:Z_1}\\ 
        Z^A_{i,2}&:=(V_i^2-\lambda_{A,2})1_{\{V_i\in A_V\}}+(1_{\{V_i\leq \Phi^{-1}(a)\}}-a)((\Phi^{-1}(a))^2-\lambda_{A,2})+(b-1_{\{V_i\leq \Phi^{-1}(b)\}})(\lambda_{A,2}-(\Phi^{-1}(b))^2)),\label{eq:Z_2}\\
         Z^A_{i,3}&:=V_iW_i 1_{\{V_i\in A_V\}},\label{eq:Z_3}\\
        Z^A_{i,4}&:=W_i1_{\{V_i\in A_V\}},\label{eq:Z_4}\\
        Z^A_{i,5}&:=(W_i^2-1)1_{\{V_i\in A_V\}}.\label{eq:Z_5}
    \end{align}
Now, we state certain properties linked to the trimmed sums; this will be used multiple times in the arguments. 

\begin{lemma}\label{lm:representations}
For any $n\in \mathbb{N}\setminus \{0\}$, we get
    \begin{align}
       \textstyle \frac{\sqrt{n}}{m_n}\sum_{i=[na]+1}^{[nb]}\left(V_{(i)}-\lambda_{A,1}\right) &\textstyle = \frac{\sqrt{n}}{m_n}\sum_{i=1}^{n}Z^A_{i,1}+E^A_{n,1},\label{eq:trim_V}\\
        \textstyle\frac{\sqrt{n}}{m_n}\sum_{i=[na]+1}^{[nb]}\left(V_{(i)}^2-\lambda_{A,2}\right) & \textstyle= \frac{\sqrt{n}}{m_n}\sum_{i=1}^{n}Z^A_{i,2}+E^A_{n,2},\label{eq:trim_V^2}\\
        \textstyle\frac{\sqrt{n}}{m_n}\sum_{i=[na]+1}^{[nb]}V_{(i)}W_{(i:n)} &\textstyle=\frac{\sqrt{n}}{m_n}\sum_{i=1}^{n}Z^A_{i,3}+E^A_{n,3},\label{eq:trim_VW}\\
        \textstyle\frac{\sqrt{n}}{m_n}\sum_{i=[na]+1}^{[nb]}W_{(i:n)}&\textstyle =\frac{\sqrt{n}}{m_n}\sum_{i=1}^nZ^A_{i,4}+E^A_{n,4},\label{eq:trim_W}\\
        \textstyle\frac{\sqrt{n}}{m_n}\sum_{i=[na]+1}^{[nb]}\left(W^2_{(i:n)}-1\right)& \textstyle=\frac{\sqrt{n}}{m_n}\sum_{i=1}^nZ^A_{i,5}+E^A_{n,5},\label{eq:trim_W^2}
        \end{align}
    where, for any  $k=1,\ldots, 5$, the sequences $(E^A_{n,k})$ are such that $E^A_{n,k}\overset{\bP}{\to} 0$ as $n\to\infty$. Moreover, for $k=1, \ldots, 5$, the sequences $(Z^A_{i,k})_{i\in \mathbb{N}}$ are i.i.d. with $\bE\left[Z^A_{i,k}\right]=0$.
\end{lemma}
\begin{proof}
    First, note that the representation stated in~\eqref{eq:trim_V} follows directly from the proof of Lemma 2 in~\cite{JelPit2021}; see the formulae following Equation (18) therein for details. Second, using the same argument applied to the squares, we also get~\eqref{eq:trim_V^2}.

    Let us now prove~\eqref{eq:trim_VW}. To ease the notation, for any $n\in \mathbb{N}\setminus \{0\}$, we define $A_n:=\sum_{i=1}^n 1_{\{V_i\leq \Phi^{-1}(a)\}}$ and $B_n:=\sum_{i=1}^n 1_{\{V_i\leq \Phi^{-1}(b)\}}$;
    note that $A_n$ and $B_n$ simply count the number of elements in the sample $(V_1, \ldots, V_n)$ which do not exceed the quantiles $\Phi^{-1}(a)$ and $\Phi^{-1}(b)$, respectively. Thus, these random variables follow the binomial distribution $B(n,a)$ and $B(n,b)$ respectively. Moreover, they are independent of $W_{(i:n)}$. Next, note that we have
    \begin{align*}
         \frac{\sqrt{n}}{m_n}\sum_{i=[na]+1}^{[nb]}V_{(i)}W_{(i:n)} &=\frac{\sqrt{n}}{m_n}\left(\sum_{i=A_n+1}^{B_n}V_{(i)}W_{(i:n)}+\mathcal{E}_{i=[na]}^{A_n}V_{(i)}W_{(i:n)}+\mathcal{E}_{i=B_n}^{[nb]}V_{(i)}W_{(i:n)}\right).
    \end{align*}
    Using the fact  that $W_{(i:n)}$ is sorted by the ordering induced by $V_i$ and recalling~\eqref{eq:Z_3}, we have 
$\sum_{i=A_n+1}^{B_n}V_{(i)}W_{(i:n)} = \sum_{i=1}^{n}V_iW_i 1_{\{V_i\in A_V\}}= \sum_{i=1}^{n}Z_{i,3}.$ Thus, we only need to show that
    \begin{align}\label{eq:conv:E_n^3}
    E^A_{n,3}:=\frac{\sqrt{n}}{m_n}\left(\mathcal{E}_{i=[na]}^{A_n}V_{(i)}W_{(i:n)}+\mathcal{E}_{i=B_n}^{[nb]}V_{(i)}W_{(i:n)}\right)\xrightarrow{\bP} 0, \quad n\to\infty.
    \end{align}
     We have
     \begin{align*}
         \frac{\sqrt{n}}{m_n}\mathcal{E}_{i=[na]}^{A_n}V_{(i)}W_{(i:n)}=\frac{\sqrt{n}}{m_n}\mathcal{E}_{i=[na]}^{A_n}\left(V_{(i)}-\Phi^{-1}(a)\right)W_{(i:n)}+\Phi^{-1}(a)\frac{\sqrt{n}}{m_n}\mathcal{E}_{i=[na]}^{A_n}W_{(i:n)}.
     \end{align*}
     Recalling that $m_n=[nb]-[na]$ and Lemma~\ref{l:binom}, we get $\Phi^{-1}(a)\frac{\sqrt{n}}{m_n}\mathcal{E}_{i=[na]}^{A_n}W_{(i:n)}=\Phi^{-1}(a)\frac{n}{m_n} \frac{1}{\sqrt{n}}\mathcal{E}_{i=[na]}^{A_n}W_{(i:n)}\xrightarrow{\bP}0$, so it is enough to show that
     \begin{align}\label{eq:38}
         \frac{\sqrt{n}}{m_n}\mathcal{E}_{i=[na]}^{A_n}\left(V_{(i)}-\Phi^{-1}(a)\right)W_{(i:n)}\xrightarrow{\bP}0.
     \end{align}
     Next, note that the sequence $W_{(i:n)}$ is i.i.d $N(0,1)$, so $\left|W_{(i:n)}\right|$ follow half-normal distribution with $\bE\left[\left|W_{(i:n)}\right|\right]=\sqrt{2 / \pi}$ and  bounded second moment. 
     By analogy to~\eqref{eq:directed_sum},  for any sequence $(z_i)$ and $k,l\in \mathbb{N}$, we define $\tilde{\mathcal{E}}_{i=k}^lz_i:=1_{\{k\neq l\}}\sum_{i=\min(k,l)+1}^{\max(k,l)}z_i$.
     Then, we have
     \begin{align*}
         0\leq\left|\frac{\sqrt{n}}{m_n}\mathcal{E}_{i=[na]}^{A_n}\left(V_{(i)}-\Phi^{-1}(a)\right)W_{(i:n)}\right|& \leq \frac{\sqrt{n}}{m_n}\tilde{\mathcal{E}}_{i=[na]}^{A_n}\left|V_{(i)}-\Phi^{-1}(a)\right|\left|W_{(i:n)}\right|\\
         &\leq\frac{\sqrt{n}}{m_n}\max\left\{\left|V_{(A_n)}-\Phi^{-1}(a)\right|,\left|V_{([na])}-\Phi^{-1}(a)\right|\right\}\tilde{\mathcal{E}}_{i=[na]}^{A_n}\left|W_{(i:n)}\right|\\
         &=\frac{n}{m_n}\max\left\{\left|V_{(A_n)}-\Phi^{-1}(a)\right|,\left|V_{([na])}-\Phi^{-1}(a)\right|\right\}\times \\
         &\phantom{=}\times \frac{1}{\sqrt{n}}\left(\tilde{\mathcal{E}}_{i=[na]}^{A_n}\left(\left|W_{(i:n)}\right|-\sqrt{2 / \pi}\right)+\left|A_n-[na]\right|\sqrt{2 / \pi}\right).
     \end{align*}
    By the consistency of empirical quantiles, we get
$\left|V_{([na])}-\Phi^{-1}(a)\right|\xrightarrow{\bP}0$ and $\left|V_{(A_n)}-\Phi^{-1}(a)\right|\xrightarrow{\bP}0$ as $n\to\infty$. Also, using the argument from Lemma~\ref{l:binom} applied to $\tilde{\mathcal{E}}_{i=k}^l$, we obtain $\frac{1}{\sqrt{n}}\tilde{\mathcal{E}}_{i=[na]}^{A_n}\left(\left|W_{(i:n)}\right|-\sqrt{\frac{2}{\pi}}\right)\xrightarrow{\bP}0$.
  Next, due to the central limit theorem and the continuous mapping theorem $\left|\frac{A_n-[na]}{\sqrt{n}}\right|$ converges in distribution to the non-degenerate folded normal random variable, so by Slutsky's theorem, we get $
        \sqrt{2/\pi}\max\left\{\left|V_{(A_n)}-\Phi^{-1}(a)\right|,\left|V_{([na])}-\Phi^{-1}(a)\right|\right\}\frac{|A_n-[na]|}{\sqrt{n}}\xrightarrow{\bP}0$, 
        which shows~\eqref{eq:38}.  Using a similar argument, we may show $        \frac{\sqrt{n}}{m_n}\mathcal{E}_{i=B_n}^{[nb]}V_{(i)}W_{(i:n)}\xrightarrow{\bP}0$,
    which together with~\eqref{eq:conv:E_n^3} concludes the proof of \eqref{eq:trim_VW}.
    
    Finally, we show~\eqref{eq:trim_W}; the proof of~\eqref{eq:trim_W^2} follows the same logic and is omitted for brevity.  First, note that
    \begin{align*}
        \frac{\sqrt{n}}{m_n}\sum_{i=[na]+1}^{[nb]}W_{(i:n)}&=\frac{\sqrt{n}}{m_n}\sum_{i=A_n+1}^{B_n}W_{(i:n)} + \frac{\sqrt{n}}{m_n}\mathcal{E}_{i=[na]}^{A_n}W_{(i:n)} + \frac{\sqrt{n}}{m_n}\mathcal{E}_{i=B_n}^{[nb]}W_{(i:n)}.
    \end{align*}
    Then, noting that $      \frac{\sqrt{n}}{m_n}\sum_{i=A_n+1}^{B_n}W_{(i:n)}=\frac{\sqrt{n}}{m_n}\sum_{i=1}^{n}W_i1_{\{V_i\in A_V\}}=\frac{\sqrt{n}}{m_n}\sum_{i=1}^{n}Z^A_{i,4}$, setting $E^A_{n,4}:=\frac{\sqrt{n}}{m_n}\mathcal{E}_{i=[na]}^{A_n}W_{(i:n)} + \frac{\sqrt{n}}{m_n}\mathcal{E}_{i=B_n}^{[nb]}W_{(i:n)}$, and using Lemma~\ref{l:binom}, we get $E^A_{n,4}\xrightarrow{\bP}0$ as $n\to\infty$, which shows of~\eqref{eq:trim_W}.  By direct computation, we get $\bE\left[Z_{i,k}^A\right]=0$ for $i\in \mathbb{N}$ and $k=1,\ldots,5$, which concludes the proof.  
\end{proof}

Recalling\eqref{eq:Z_1}--\eqref{eq:Z_5}, to ease the notation, we define 
\begin{equation}\label{eq:Z_i^A}
Z_i^A:=\left[Z^A_{i,1}, Z^A_{i,2}, Z^A_{i,3}, Z^A_{i,4}, Z^A_{i,5}\right]^\top,\quad i\in \mathbb{N}\setminus \{0\}.
\end{equation}
Then, using Lemma~\ref{lm:representations}, we get that the random vectors $Z_i^A$ are i.i.d. with zero mean, thus by the multivariate central limit theorem combined with Slutsky's theorem, we get
\begin{align}\label{eq:convergence}
    \frac{\sqrt{n}}{m_n}\sum_{i=1}^nZ^A_{i}&= \frac{n}{m_n}\sqrt{n}\frac{1}{n}\sum_{i=1}^nZ^A_{i}\xrightarrow{d}\left[J_{A,1},\,J_{A,2}, \,B_{A,1},\,B_{A,2},\,B_{A,3}\right]^\top,
\end{align}
 where $\left[J_{A,1},J_{A,2},B_{A,1},B_{A,2},B_{A,3}\right]^\top$ is a normally distributed random vector with $0$ mean and the covariance matrix $\Sigma^A$ given by $\Sigma^A(k,l)=\frac{1}{(b-a)^2} \cov\left[Z^A_{1,k},Z^A_{1,l}\right]$, $k,l=1,...,5$.

Now, we show how to use~\eqref{eq:convergence} to prove the asymptotic normality of the sample conditional mean and covariance. This could be seen as an extension of the results from Section 4 in~\cite{JelPit2021} for the bivariate case.

\begin{proposition}\label{l.mur}
Let $\hat{\mu}_{A,j}$, $\hat{r}_{A}$, $\mu_{A,j}$, and $r_A$ be given by~\eqref{eq:mu_A:hat},~\eqref{eq:r_A:hat},~\eqref{eq:mu_A}, and~\eqref{eq:r_A}, respectively. Then, we get
\begin{align}
\sqrt{n} ( \hat{\mu}_{A,j} -\mu_{A,j})&\xrightarrow{d}  
\sigma_j \left( \cos(\gamma_j) J_{A,1} + \sin(\gamma_j) B_{A,2}\right),\label{eq:l.mur:mu} \\
\sqrt{n} \left( \hat{r}_{A} - r_A\right)
&\xrightarrow{d}
\sigma_1\sigma_2 \left(\cos(\gamma_1) \cos(\gamma_2)  (J_{A,2}- 2\lambda_{A,1} J_{A,1})+\sin(\gamma_1+\gamma_2) \left(B_{A,1}-\lambda_{A,1}B_{A,2}\right)
+ \sin(\gamma_1)\sin(\gamma_2) B_{A,3}\right),\label{eq:l.mur:r}
\end{align}
where $J_{A,1},J_{A,2},B_{A,1},B_{A,2},B_{A,3}$ are given as in \eqref{eq:convergence}. In particular, the limiting random variables are normally distributed. 
\end{proposition}
\begin{proof}

First, we show~\eqref{eq:l.mur:mu}. Recalling Lemma~\ref{lm:representations} and \eqref{eq:convergence}, by the multivariate central limit theorem combined with Slutsky's theorem, we get 
\begin{align*}
\sqrt{n} ( \hat{\mu}_{A,j} -\mu_{A,j})&=\sqrt{n} ( \hat{\mu}_{A,j} - \mu_j - \sigma_j \cos(\gamma_j) \lambda_{A,1})\nonumber\\
&=\frac{\sigma_j\sqrt{n}}{m_n} \left( \cos(\gamma_j) \sum_{i=[na]+1}^{[nb]} V_{(i)} 
+ \sin(\gamma_j) \sum_{i=[na]+1}^{[nb]} W_{(i:n)}-m_n\cos(\gamma_j) \lambda_{A,1}\right)\nonumber\\
&=\sigma_j\left(\cos(\gamma_j)\frac{\sqrt{n}}{m_n}\sum_{i=[na]+1}^{[nb]} \left(V_{(i)}-\lambda_{A,1}\right)+ \sin(\gamma_j)\frac{\sqrt{n}}{m_n}\sum_{i=[na]+1}^{[nb]} W_{(i:n)}\right)\nonumber\\
& = \sigma_j\left(\cos(\gamma_j)\frac{\sqrt{n}}{m_n}\sum_{i=1}^n Z^A_{i,1}+ \sin(\gamma_j)\frac{\sqrt{n}}{m_n}\sum_{i=[na]+1}^{[nb]} Z^A_{i,4}\right)+\sigma_j \cos (\gamma_j)E^A_{n,1}+\sigma_j \sin (\gamma_j)E^A_{n,4}\nonumber\\
&= \sigma_j\begin{bmatrix}
    \cos(\gamma_j), & \sin(\gamma_j)
\end{bmatrix}\begin{bmatrix}
    \frac{\sqrt{n}}{m_n}\displaystyle\sum_{i=1}^n Z^A_{i,1} \\
    \frac{\sqrt{n}}{m_n}\displaystyle\sum_{i=1}^n Z^A_{i,4}
\end{bmatrix} + \sigma_j\cos (\gamma_j)E^A_{n,1} + \sigma_j \sin (\gamma_j)E^A_{n,4}\nonumber\\
&\xrightarrow{d} \sigma_j \left( \cos(\gamma_j) J_{A,1} + \sin(\gamma_j) B_{A,2}\right),
\end{align*}
which concludes the proof of~\eqref{eq:l.mur:mu}.

\indent Now, we show \eqref{eq:l.mur:r}. First, note that 
 from~\eqref{eq:trim_V}, by direct computation, we get
\begin{align*}
    \sqrt{n}\left(\left(\frac{1}{m_n}\sum_{i=[na]+1}^{[nb]}V_{(i)}\right)^2-(\lambda_{A,1})^2\right) = 2\lambda_{A,1}\frac{\sqrt{n}}{m_n}\sum_{i=1}^n Z^A_{i,1} +\sqrt{n}\left(\frac{1}{m_n}\sum_{i=1}^n Z^A_{i,1}\right)^2 + \frac{1}{\sqrt{n}}\left(E^A_{n,1}\right)^2 +  2\lambda_{A,1}E^A_{n,1} + \frac{2}{m_n}E^A_{n,1}\sum_{i=1}^{n}Z_{i,1}^A.
\end{align*}
Also, multiplying both hand sides of~\eqref{eq:trim_V} by $\frac{1}{\sqrt{n}}$ and using the law of large numbers, we obtain
\begin{align}\label{eq:WLLN}
    \frac{1}{m_n}\sum_{i=[na]+1}^{[nb]}V_{(i)}=\frac{1}{m_n}\sum_{i=1}^n Z^A_{i,1}+\frac{1}{\sqrt{n}}E^A_{n,1}+\lambda_{A,1}\xrightarrow{\bP}\lambda_{A,1}.       
\end{align}
Then, recalling~\eqref{eq:Z_1}--\eqref{eq:Z_5}, we get
\begingroup
\allowdisplaybreaks
\begin{align}\label{eq:l.mur:r:1}
   \sqrt{n} \left( \hat{r}_{A} - r_A\right)&= \sqrt{n} \left( \hat{r}_{A} - \sigma_1 \sigma_2 \left( \cos(\gamma_1) \cos(\gamma_2) (\lambda_{A,2}- (\lambda_{A,1})^2) + \sin(\gamma_1)\sin(\gamma_2) \right)\right)\nonumber\\
&=\sqrt{n}\left(\frac{\sigma_1 \sigma_2}{m_n}
  \left(\cos(\gamma_1) \cos(\gamma_2)  \sum_{i=[na]+1}^{[nb]}V_{(i)}^2 
+\sin(\gamma_1+\gamma_2)  \sum_{i=[na]+1}^{[nb]}V_{(i)} W_{(i:n)}+
  \sin(\gamma_1)\sin(\gamma_2)\sum_{i=[na]+1}^{[nb]}W_{(i:n)}^2 \right)\right.\nonumber\\
&\phantom{=\sqrt{n}}- \frac{\sigma_1\sigma_2}{m^2_n}{}\left(\cos(\gamma_1) \sum_{i=[na]+1}^{[nb]} V_{(i)} 
+ \sin(\gamma_1) \sum_{i=[na]+1}^{[nb]} W_{(i:n)}\right)\left(\cos(\gamma_2) \sum_{i=[na]+1}^{[nb]} V_{(i)} 
+ \sin(\gamma_2) \sum_{i=[na]+1}^{[nb]} W_{(i:n)}\right)\nonumber\\
&\phantom{=\sqrt{n}}\left.- \sigma_1 \sigma_2 \left( \cos(\gamma_1) \cos(\gamma_2) (\lambda_{A,2}- (\lambda_{A,1})^2) + \sin(\gamma_1)\sin(\gamma_2) \right)\right)\nonumber\\
&=\sigma_1\sigma_2\left(\cos(\gamma_1)\cos(\gamma_2)\left(\frac{\sqrt{n}}{m_n}\sum_{i=[na]+1}^{[nb]}\left(V^2_{(i)}-\lambda_{A,2}\right)-\sqrt{n}\left(\left(\frac{1}{m_n}\sum_{i=[na]+1}^{[nb]}V_{(i)}\right)^2-(\lambda_{A,1})^2\right)\right)\right.\nonumber\\
&\phantom{=\sigma_1\sigma_2}+\sin(\gamma_1+\gamma_2)\left(\frac{\sqrt{n}}{m_n}\sum_{i=[na]+1}^{[nb]}V_{(i)}W_{(i:n)}-\frac{\sqrt{n}}{m_n^2}\sum_{i=[na]+1}^{[nb]}V_{(i)}\sum_{i=[na]+1}^{[nb]}W_{(i:n)}\right)\nonumber\\
&\phantom{=\sigma_1\sigma_2}\left.+\sin(\gamma_1)\sin(\gamma_2)\left(\frac{\sqrt{n}}{m_n}\sum_{i=[na]+1}^{[nb]}\left(W^2_{(i:n)}-1\right)-\frac{\sqrt{n}}{m_n^2}\left(\sum_{i=[na]+1}^{[nb]}W_{(i:n)}\right)^2\right)\right)\nonumber\\
&=\sigma_1\sigma_2\left(\cos(\gamma_1)\cos(\gamma_2)\left(\frac{\sqrt{n}}{m_n}\sum_{i=1}^{n}Z^A_{i,2}-2\lambda_{A,1}\frac{\sqrt{n}}{m_n}\sum_{i=1}^{n}Z^A_{i,1}\right)\right.\nonumber\\
&\phantom{=\sigma_1\sigma_2}\left.+\sin(\gamma_1+\gamma_2)\left(\frac{\sqrt{n}}{m_n}\sum_{i=1}^{n}Z^A_{i,3}-\frac{1}{m_n}\sum_{i=[na]+1}^{[nb]}V_{(i)}\frac{\sqrt{n}}{m_n}\sum_{i=1}^{n}Z^A_{i,4}\right)+\sin(\gamma_1)\sin(\gamma_2)\frac{\sqrt{n}}{m_n}\sum_{i=1}^{n}Z^A_{i,5}\right)+H^A_n\nonumber\\
&=M^A_n \frac{\sqrt{n}}{m_n}\sum_{i=1}^nZ^A_{i}+H^A_n,
\end{align}
\endgroup
where we set
\begin{align}
M^A_n&:=\sigma_1\sigma_2\left[
    -2\lambda_{A,1}\cos(\gamma_1)\cos(\gamma_2),\, \cos(\gamma_1)\cos(\gamma_2),\, \sin(\gamma_1+\gamma_2),\, -\frac{\sin(\gamma_1+\gamma_2)}{m_n}\sum_{i=[na]+1}^{[nb]}V_{(i)},\, \sin(\gamma_1)\sin(\gamma_2)
\right]\label{eq:M_n}\\
\label{eq:H_n}
    H^A_n&:=\sigma_1\sigma_2\left(\cos(\gamma_1)\cos(\gamma_2)\left(E^A_{n,2}+\sqrt{n}\left(\frac{1}{m_n}\sum_{i=1}^{n}Z_{i,1}^A\right)^2+ E^A_{n,1}\left(\frac{1}{\sqrt{n}}E^A_{n,1}+2\lambda_{A,1}+\frac{2}{m_n}\sum_{i=1}^{n}Z_{i,1}^A\right)\right)\right.\nonumber\\
    &\phantom{=}\left.+\sin(\gamma_1+\gamma_2)\left(E^A_{n,3}-\frac{1}{m_n}\sum_{i=[na]+1}^{[nb]}V_{(i)}E^A_{n,4}\right)+\sin(\gamma_1)\sin(\gamma_2)\left(E^A_{n,5}+\frac{\sqrt{n}}{m_n^2}\left(\sum_{i=[na]+1}^{[nb]}W_{(i:n)}\right)^2\right)\right).
\end{align}
By Lemma~\ref{lm:representations} and \eqref{eq:WLLN}, we have $H_n^A\xrightarrow{\bP}0$ and
\begin{align*}
    M^A_n\xrightarrow{\bP}\sigma_1\sigma_2\left[
    -2\lambda_{A,1}\cos(\gamma_1)\cos(\gamma_2),\, \cos(\gamma_1)\cos(\gamma_2),\, \sin(\gamma_1+\gamma_2),\, -\sin(\gamma_1+\gamma_2)\lambda_{A,1},\, \sin(\gamma_1)\sin(\gamma_2)
\right].
\end{align*}
Consequently, combining \eqref{eq:l.mur:r:1} with Slutsky's Theorem and \eqref{eq:convergence}, we obtain
\begin{align}\label{eq:conv_r_A}
    \sqrt{n}\left(\hat{r}_{A}-r_A\right)&=M^A_n \frac{\sqrt{n}}{m_n}\sum_{i=1}^nZ^A_{i}+H^A_n\nonumber\\
    &\xrightarrow{d}\sigma_1\sigma_2 \left(\cos(\gamma_1) \cos(\gamma_2)  (J_{A,2}- 2\lambda_{A,1} J_{A,1})+\sin(\gamma_1+\gamma_2) \left(B_{A,1}-\lambda_{A,1}B_{A,2}\right)
+ \sin(\gamma_1)\sin(\gamma_2) B_{A,3}\right),
\end{align}
which concludes the proof.
\end{proof}

From Proposition~\ref{l.mur} we conclude that the estimators $\hat{\mu}_{A,j}$ and $\hat{r}_{A}$ are asymptotically normal. Thus, as the asymptotic normality implies the consistency of the estimator, we get the following corollary.
\begin{corollary}
    We have $\hat{\mu}_{A,j}\xrightarrow{\mathbb{P}} \mu_{A,j}$ and $\hat{r}_{A}\xrightarrow{\mathbb{P}}r_A$.
\end{corollary}

Finally, we prove Theorem~\ref{theorem}.

\begin{proof}[{\bf Proof of Theorem~\ref{theorem}}]
In the proof we focus on $T_n$, the arguments for $L_n$ and $R_n$ are similar and we omit them for brevity.

First, let us define $T_n^\prime:=\sqrt{n}\left(\hat{r}_{A_1}+\hat{r}_{A_3}-2\hat{r}_{A_2}\right)$ and show that
\begin{align}\label{eq:52}
    T_n^\prime\xrightarrow{d}T_\infty&:=\sigma_1\sigma_2 
    \left(\cos(\gamma_1) \cos(\gamma_2)  \left(J_{A_1,2}-2J_{A_2,2}+J_{A_3,2}- 2\lambda_{A_1,1}J_{A_1,1} - 2\lambda_{A_3,1} J_{A_3,1}\right)\right. \nonumber\\
    &\phantom{=\sigma_1\sigma_2}+\sin(\gamma_1+\gamma_2) \left(B_{A_1,1}-\lambda_{A_1,1}B_{A_1,2}-2B_{A_2,1}+B_{A_3,1}-\lambda_{A_3,1}B_{A_3,2}\right) \nonumber\\
    &\phantom{=\sigma_1\sigma_2}\left.+\sin(\gamma_1)\sin(\gamma_2)  \left(B_{A_1,3}-2B_{A_2,3}+B_{A_3,3}\right)\right),
\end{align}
    where $B_{A_k,i}$ and $J_{A_k}^j$ are as in \eqref{eq:convergence}. In particular, this implies that $T_\infty$ is normally distributed as by~\eqref{eq:convergence} it is a linear combination of the multivariate normal random variables. To see that~\eqref{eq:52} holds, we recall that by Theorem~\ref{Th:206020}, we have $r_{A_1}=r_{A_2}=r_{A_3}$.
Thus, as in the proof of the Proposition~\ref{l.mur}, we obtain
\begin{align*}
    T_n^\prime&=\sqrt{n}\left(\hat{r}_{A_1}-r_{A_1} +\hat{r}_{A_3}-r_{A_3} -2\hat{r}_{A_2}+2r_{A_2}\right)=\left[M^{A_1}_n,\, -2M^{A_2}_n,\, M^{A_3}_n\right] \frac{\sqrt{n}}{m_n}\sum_{i=1}^n\begin{bmatrix}
           Z^{A_1}_i\\ Z^{A_2}_i\\ Z^{A_3}_i 
        \end{bmatrix} +H^{A_1}_n-2H^{A_2}_n+H^{A_3}_n,  
\end{align*}
where $Z_i^{A_k} $, $M_n^{A_k}$, and $ H_n^{A_k}$ are as in~\eqref{eq:Z_i^A}, \eqref{eq:M_n}, and \eqref{eq:H_n}, respectively. Then, recalling the argument leading to~\eqref{eq:conv_r_A} and using the multivariate central limit theorem combined with Slutsky's theorem, we get $T_n'\xrightarrow{d}T_\infty$ as $n\to\infty$, which shows~\eqref{eq:52}.

   Second, we provide a formula for the variance of $T_\infty$. By direct computation, we get
        \begin{align}\label{eq:th:final:variance}
    \Var(T_\infty)=\sigma_1^2\sigma_2^2\left(\cos^2(\gamma_1)\cos^2(\gamma_2)C_1+\sin^2(\gamma_1+\gamma_2)C_2+\sin^2(\gamma_1)\sin^2(\gamma_2)C_3\right),
    \end{align}
    where we have
    \begin{align*}
        C_1:&=\Var\left[J_{A_1,2}-2J_{A_2,2}+J_{A_3,2}- 2\lambda_{A_1,1} J_{A_1,1} - 2\lambda_{A_3,1} J_{A_3,1}\right]\\
        &=2\kappa_{A_1}+4\kappa_{A_2}-8\lambda_{A_1,1}\xi_{A_1}+2\tilde{q}(\lambda_{A_1,2})^2 + 2\tilde q(1-\tilde q)\left(\left(\tilde M_{1,2,2}\right)^2+4\left(\tilde M_{2,2,2}\right)^2\right) + 2\left(\tilde q\tilde M_{1,2,2}\right)^2+4\tilde q(\lambda_{A_2,2})^2\\
        &\phantom{=}+8\tilde M_{2,2,2} \left(\tilde M_{2,2,1}\right)^2\left(\lambda_{A_2,2}\left(2\tilde q  -\lambda_{A_2,2}\right)-\tilde q\right) + 8(\lambda_{A_1,1})^2\left(2\tilde q \lambda_{A_1,2}-\tilde q (\lambda_{A_1,1})^2+\tilde q(1-\tilde q)\tilde M_{1,1,1}\right)\\
        &\phantom{=} -8\lambda_{A_1,1}\left(\tilde q(1-\tilde q)\tilde M_{1,1,1}\tilde M_{1,2,2}+\tilde a^2 \left(M_{1,2,1}\right)^2\tilde M_{1,1,1}\right)\\
        &\phantom{=}+ 16\lambda_{A_1,1}\tilde N_1\left((1-2\tilde q)\left((1-\tilde q)\lambda_{A_1,1} + \tilde q \lambda_{A_2,2}-\tilde q-\lambda_{A_1,1}\lambda_{A_2,2}\right) +\tilde N_2(1-2\tilde q)\lambda_{A_1,2} \right)\\  
        C_2:&=\Var\left[B_{A_1,1}-\lambda_{A_1,1}B_{A_1,2}-2B_{A_2,1}+B_{A_3,1}-\lambda_{A_3,1}B_{A_3,2}\right]=2\frac{\lambda_{A_1,2}-(\lambda_{A_1,1})^2}{\tilde{q}}+4\frac{\lambda_{A_2,2}}{1-2\tilde{q}},\\
        C_3:&=\Var\left[B_{A_1,3} + B_{A_3,3} -2B_{A_2,3} \right]=\frac{4}{\tilde{q}}+\frac{8}{1-2\tilde{q}},
    \end{align*}
    with $\kappa_{A_i}:=\bE\left[V_1^4\mid V_1\in A_1\right],\, \xi_{A_i}:=\bE\left[V_1^3\mid V_1\in A_1\right],\, \tilde a := \Phi^{-1}(\tilde q),\, \tilde M_{i,j,k} := \tilde a^k -\lambda_{A_i,j},\, \tilde N_k:=\tilde q-\tilde a^k$.
By numerical evaluation we obtain $C_1\approx3.1989$, $C_2\approx3.6412$, and $C_3\approx33.4424$. Now, we reexpress~\eqref{eq:th:final:variance} in terms of the parameters of the underlying random vector $(X_{1,1},X_{1,2},Y)$. Recalling~~\eqref{eq:X_params}, by extensive computations, we obtain 
\begin{align*}
    \sigma_1^2\sigma_2^2\cos^2(\gamma_1)\cos^2(\gamma_2)&=\frac{\cov^2(Y_1,X_{1,1})\cov^2(Y_1,X_{1,2})}{\sigma^4(Y_1)},\\
    \sigma_1^2\sigma_2^2\sin^2(\gamma_1+\gamma_2)&=\sigma_1^2\sigma_2^2(\cos^2(\gamma_1)+\cos^2(\gamma_2)+2\cos(\gamma_1-\gamma_2)\cos(\gamma_1)\cos(\gamma_2)-4\cos^2(\gamma_1)\cos^2(\gamma_2))\nonumber\\
        &=\frac{\sigma_2^2\cov^2(Y_1,X_{1,1})+\sigma_1^2\cov^2(Y_1,X_{1,2})+2\cov(X_{1,1},X_{1,2})\cov(Y_1,X_{1,1})\cov(Y_1,X_{1,2})}{\sigma^2(Y_1)}\nonumber\\
        &\phantom{=}-4\frac{\cov^2(Y_1,X_{1,1})\cov^2(Y_1,X_{1,2})}{\sigma^4(Y_1)},\\   
\sigma_1^2\sigma_2^2\sin^2(\gamma_1)\sin^2(\gamma_2)&=\left(\sigma_1^2-\frac{\cov^2(Y_1,X_{1,1})}{\sigma^2(Y_1)}\right)\left(\sigma_2^2-\frac{\cov^2(Y_1,X_{1,2})}{\sigma^2(Y_1)}\right).
    \end{align*}
Plugging these into~\eqref{eq:th:final:variance}, we get
\begingroup
\allowdisplaybreaks
    \begin{align*}
        \Var(T_\infty)&=\frac{\cov^2(Y_1,X_{1,1})\cov^2(Y_1,X_{1,2})}{\sigma^2(Y_1)}C_1\\
        &\phantom{=}+\Bigg(\frac{\sigma_2^2\cov^2(Y_1,X_{1,1})+\sigma_1^2\cov^2(Y_1,X_{1,2})+2\cov(X_{1,1},X_{1,2})\cov(Y_1,X_{1,1})\cov(Y_1,X_{1,2})}{\sigma^2(Y_1)}\\
        &\phantom{=}-4\frac{\cov^2(Y_1,X_{1,1})\cov^2(Y_1,X_{1,2})}{\sigma^4(Y_1)}\Bigg)C_2+\left(\sigma_1^2-\frac{\cov^2(Y_1,X_{1,1})}{\sigma^2(Y_1)}\right)\left(\sigma_2^2-\frac{\cov^2(Y_1,X_{1,2})}{\sigma^2(Y_1)}\right)C_3\\
        &=\frac{\cov^2(Y_1,X_{1,1})\cov^2(Y_1,X_{1,2})}{\sigma^4(Y_1)}(C_1-4C_2+C_3)\\
        &\phantom{=}+\left(\frac{\sigma_2^2\cov^2(Y_1,X_{1,1})+2\cov(X_{1,1},X_{1,2})\cov(Y_1,X_{1,1})\cov(Y_1,X_{1,2})+\sigma_1^2\cov^2(Y_1,X_{1,2})}{\sigma^2(Y_1)}\right)\left(C_2-C_3\right)\\
        &\phantom{=}+\left(\sigma_1^2\sigma_2^2+2\frac{\cov(X_{1,1},X_{1,2})\cov(Y_1,X_{1,1})\cov(Y_1,X_{1,2})}{\sigma^2(Y_1)}\right)C_3.
    \end{align*}
\endgroup
 Next, we set $K_1:=C_1-4C_2+C_3\approx 22.0766$, $K_2:=C_2-C_3\approx -29.8012$, $K_3:=C_3 \approx 33.4424$. Also, we recall $\hat{\tau}^2$ from \eqref{c.tau} and note that this is a consistent estimator of $\Var(T_\infty)$, i.e we have $\frac{1}{\hat{\tau}}\sqrt{\Var(T_\infty)}\xrightarrow{\mathbb{P}}1$. Thus, from Slutsky's theorem combined with~\eqref{eq:52}, we obtain
    \begin{align*}
        T_n=\frac{1}{\hat{\tau}}T_n^\prime=\frac{\sqrt{\Var(T_\infty)}}{\hat{\tau}}\frac{T_n^\prime}{\sqrt{\Var(T_\infty)}}\xrightarrow{d}\frac{T_{\infty}}{\sqrt{\Var(T_\infty)}}\sim N(0,1),
    \end{align*}
    which concludes the proof.
\end{proof}
\begin{remark}\label{rm:constants2}
    By the same argument, we can prove that statistics $L_n$ and $R_n$ defined in~\eqref{T1.stat} converge asymptotically to standard normal distribution. Also, due to the symmetry, the normalizing statistic given by~\eqref{c.delta} is common for both $R_n$ and $L_n$ and we have $\tilde K_1:=\tilde C_1-4\tilde C_2+\tilde C_3$, $\tilde K_2:=\tilde C_2-\tilde C_3$, $\tilde K_3:=\tilde C_3$ with 
    \begin{align*}
        \tilde C_1:&=\Var\left[J_{A_1,2}-J_{A_2,2}- 2\lambda_{A_1,1} J_{A_1,1}\right]\\
        &=\kappa_{A_1}+\kappa_{A_2}-4\lambda_{A_1,1}\xi_{A_1}+\tilde{q}\left((\lambda_{A_1,2})^2+(\lambda_{A_2,2})^2\right) + \tilde q(1-\tilde q)\left(\left(\tilde M_{1,2,2}\right)^2+2\left(\tilde M_{2,2,2}\right)^2\right)\\
        &\phantom{=}+2\tilde M_{2,2,2} \left(\tilde M_{2,2,1}\right)^2\left(\lambda_{A_2,2}\left(2\tilde q  -\lambda_{A_2,2}\right)-\tilde q\right)+4(\lambda_{A_1,1})^2\left(2\tilde q \lambda_{A_1,2}-\tilde q (\lambda_{A_1,1})^2+\tilde q(1-\tilde q)\tilde M_{1,1,1}\right)\\
        &\phantom{=} +4\lambda_{A_1,1}\left(\tilde q(1-\tilde q)\tilde M_{1,1,1}\tilde M_{1,2,2}\right)+4\lambda_{A_1,1}\tilde N_1\left((1-2\tilde q)\left((1-\tilde q)\lambda_{A_1,1} + \tilde q \lambda_{A_2,2}-\tilde q-\lambda_{A_1,1}\lambda_{A_2,2}\right) +\tilde N_2(1-2\tilde q) \right)\\
        &\phantom{=}-2\tilde N_2\left((1-2\tilde q)\left(\tilde q+\lambda_{A_1,1}\lambda_{A_2,2}-(1-\tilde q)\lambda_{A_1,1} - \tilde q \lambda_{A_2,2}\right)+\tilde N_2(1-2\tilde q)\lambda_{A_1,2}\right)\approx 1.2792,\\
        \tilde C_2:&=\Var\left[B_{A_1,1}-\lambda_{A_1,1}B_{A_1,2}-B_{A_2,1}\right]=\frac{\lambda_{A_1,2}-(\lambda_{A_1,1})^2}{\tilde{q}}+\frac{\lambda_{A_2,2}}{1-2\tilde{q}}\approx 1.4599,\\
        \tilde C_3:&=\Var\left[B_{A_1,3}-B_{A_2,3}\right]= 2\frac{1-\tilde q}{\tilde{q}(1-2\tilde q)}\approx 13.4091.
    \end{align*}
\end{remark}

\section*{Acknowledgements}
 Marcin Pitera, Jakub Wo\'zny, and Agnieszka Wy\l oma\'nska acknowledge support from the National Science Centre, Poland, via project 2020/37/B/HS4/00120. Damian Jelito acknowledge support from the National Science Centre, Poland, via project 2020/37/B/ST1/00463. Part of the work of Jakub Woźny was funded by the program Excellence Initiative – Research University at the Jagiellonian University in Kraków.

\bibliography{bibliography}

\end{document}